\documentclass[letterpaper,12pt]{amsart}
\usepackage{amsmath,amsfonts,amsthm,mathrsfs,amssymb,graphicx,subfigure,bbm}
\usepackage[margin=1in]{geometry}

\usepackage{color}

\definecolor{myblue}{rgb}{0.3, 0.0, 0.85}
\definecolor{myviolet}{rgb}{0.5, 0.0, 0.5}

\usepackage{hyperref}

\hypersetup{
	colorlinks=true,
	linkcolor=myblue,
	citecolor=myviolet
}

\newcommand\set{\mathrel{\overset{\makebox[0pt]{\mbox{\normalfont\tiny\sffamily \text{set}}}}{=}}}

\DeclareMathOperator*{\slim}{s-lim}
\DeclareMathOperator*{\wlim}{w-lim}

\theoremstyle{plain}
\newtheorem{thm}{Theorem}[section]
\newtheorem{lem}{Lemma}[section]

\newtheorem{cor}{Corollary}[section]
\newtheorem{defn}{Definition}[section]

\newtheorem*{rem}{Remark}

\newtheorem*{unthm}{Theorem}
\newtheorem*{unlem}{Lemma}

\newtheorem{snthm}{Theorem}
\newtheorem{snprop}{Proposition}

\title{Dynamics of Noncommutative Solitons I: Spectral Theory and Dispersive Estimates}
\author{August J. Krueger}
\address{Amado Building, Technion Math.\ Dept., Haifa, 32000, Israel}
\email{ajkrueger@tx.technion.ac.il}
\author{Avy Soffer}
\address{110 Frelinghuysen Rd., Rutgers Math.\ Dept., Piscataway, NJ 08854, USA}
\email{soffer@math.rutgers.edu}

\begin{document}

\maketitle


\begin{abstract}
We consider the Schr\"odinger equation with a Hamiltonian given by a second order difference operator with nonconstant growing coefficients, on the half one dimensional lattice. This operator appeared first naturally in the construction and dynamics of noncommutative solitons in the context of noncommutative field theory. We prove pointwise in time decay estimates with the decay rate $t^{-1}\log^{-2}t$, which is optimal with the chosen weights and appears to be so generally. We use a novel technique involving generating functions of orthogonal polynomials to achieve this estimate.
\end{abstract}

\section{Introduction and Background}

The notion of noncommutative soliton arises when one considers the nonlinear Klein-Gordon equation (NLKG) for a field which is dependent on, for example, two ``noncommutative coordinates'', $x,y$, whose coordinate functions satisfy canonical commutation relations (CCR) $[X,Y]=i\epsilon$. This follows through the method of deformation quantization, see e.g. \cite{BaFl1} for a review and \cite{BaFl2} for applications. By going to a representation of the above canonical commutation relation, one can reduce the dynamics of the problem to an equation for the coefficients of an expansion in the Hilbert space representation of the above CCR, see e.g. \cite{DJN 1}\cite{DJN 2}\cite{GMS}. By restricting to rotationally symmetric functions the nocommutative deformation of the Laplacian reduces to a second order finite difference operator, which is symmetric, and with variable coefficient growing like the lattice coordinate, at infinity. Therefore, this operator is unbounded, and in fact has continuous spectrum $[0,\infty)$. These preliminary analytical results, as well as additional numerical results, were obtained by Chen, Fr\"ohlich, and Walcher \cite{CFW}. The dynamics and scattering of the (perturbed) soliton can then be inferred from the NLKG with such a discrete operator as the linear part. We will be interested  in studying the dynamics of discrete NLKG and discrete NLS equations with these hamiltonians.


We will be working with a discrete Schr\"odinger operator $L_0$ which can be considered either a discretization or a noncommutative deformation of the radial 2D negative Laplacian, $-\Delta^{\mathrm{2D}}_\mathrm{r} = -r^{-1}\partial_r r \partial_r$. We will briefly review both perspectives.

In 1D one may find a discrete Laplacian via
\begin{align*}
& x \in \mathbb{R} \  \xrightarrow{\ \mathrm{discrete} \ } \  n \in \mathbb{Z}, \quad -\Delta^{\mathrm{1D}} = -\partial^{2}_x \  \xrightarrow{\ \mathrm{discrete} \ } \  -D_{+}D_{-},
\end{align*}
where $D_{+}v(n) = v(n+1) - v(n), D_{-}v(n) = v(n) - v(n-1)$ are respectively the forward and backward finite difference operators. It is important to implement this particular combination of these finite difference operators due in order to ensure that the resulting discrete Laplacian is symmetric. In 2D one may find a discrete Laplacian via
\begin{align*}
&r = (x^{2} + y^{2})^{1/2} =  2 \rho^{1/2},\quad \rho \in \mathbb{R}_+ \  \xrightarrow{\ \mathrm{discrete} \ } \  n \in \mathbb{Z}_+, \\
&-\Delta^{\mathrm{2D}}_{\mathrm{r}} = -r^{-1}\partial_{r}r\partial_{r} = -\partial_\rho \rho \partial_\rho \quad \xrightarrow{\ \mathrm{discrete} \ } \quad -D_{+}MD_{-} = L_{0},
\end{align*}
where $Mv(n) = nv(n)$. For any 1D  continuous coordinate $x$ one may discretize a pointwise multiplication straightforwardly via $v^p(x) \  \xrightarrow{\ \mathrm{discrete} \ } \  v^p(n)$, where $n$ is a discrete coordinate.

One may also follow the so-called noncommutative space perspective. Here one considers the formal ''Moyal star deformation'' of the algebra of functions on $\mathbb{R}^{2}$:
\begin{align*}
\Phi_{1}\cdot\Phi_{2}(x,y) &= \Phi_{1}(x,y)\Phi_{2}(x,y) \\
\xrightarrow{\ \epsilon > 0\ } \quad \Phi_{1} \star \Phi_{2}(x, y) &= \exp[i(\epsilon / 2)(\partial_{x_1}\partial_{y_2} - \partial_{y_1}\partial_{x_2})] \Phi_{1}(x_1,y_1)\\
&\qquad\times\Phi_{2}(x_2,y_2)\lfloor_{ (x_j,y_j) = (x,y) }.
\end{align*}
One calls the coordinates, $x,y$, noncommutative in this context because the coordinate functions $X(x,y) = x$, $Y(x,y) = y$ satisfy a nontrivial commutation relation $X\star Y - Y\star X \equiv [X,Y] = i\epsilon$. This prescription can be considered equivalent to the multiplication of functions of $q, p$ in quantum mechanics where operator ordering ambiguities are set by the normal ordering prescription for each product. For $\Phi$ a deformed function of $r = (x^2 + y^2)^{1/2}$ alone: $\Phi = \sum_{n=0}^\infty v(n)\Phi_n$ where $v(n) \in \mathbb{C}$ and the $\{\Phi_{n}\}_{n=0}^\infty$ are distinguished functions of $r$: the projectors onto the eigenfunctions of the noncommutative space variant of quantum simple harmonic oscillator system. One may find for $\Phi$ a function of $r$ alone:
\begin{align*}
-\Delta^{\mathrm{2D}} \Phi &= -\Delta^{\mathrm{2D}}_{\mathrm{r}}\Phi = -r^{-1}\partial_{r}r\partial_{r}\Phi \\
\xrightarrow{\ \epsilon > 0\ }\quad \frac{2}{\epsilon} L_0\Phi_n &= \frac{2}{\epsilon} \left\{ \begin{array}{cc}
		- (n+1)\Phi_{n+1} + (2n + 1)\Phi_{n} - n \Phi_{n-1} &,\quad n > 0 \\
		- \Phi_{1} + \Phi_{0} &,\quad n = 0 .
	\end{array} \right.
\end{align*}
which may be transferred to $\frac{2}{\epsilon} L_0v(n)$, an equivalent action on the $v(n)$, due to the symmetry of $L_0$. Since the $\Phi_n$ are noncommutative space representations of projection operators on a standard quantum mechanical Hilbert space, they diagonalize the Moyal star product: $\Phi_m\star\Phi_n = \delta_{m,n}\Phi_n$. This property is shared by all noncommutative space representations of projection operators. Thereby products of the $\Phi_n$ may be transferred to those of the expansion coefficients: $v(n)v(n) = v^2(n)$.

See  B. Durhuus, T. Jonsson, and R. Nest \cite{DJN 1,DJN 2} (2001) and T. Chen, J. Fr\"ohlich, and J. Walcher \cite{CFW} (2003) for reviews of the two approaches. In the following we will work on a lattice explicitly so $x \in \mathbb{Z}_+$ will be a discrete spatial coordinate.

The principle of replacing the usual space with a noncommutative space (or space-time) has found extensive use for model building in physics and in particular for allowing easier construction of localized solutions, see e.g. \cite{fuzzy physics}\cite{NC soliton survey} for surveys. An example of the usefulness of this approach is that it may provide a robust procedure for circumventing classical nonexistence theorems for solitons, e.g. that of Derrick \cite{Derrick}. The NLKG variant of the equation we study here first appeared in the context of string theory and associated effective actions in the presence of background D-brane configurations, see e.g. \cite{GMS}. We have decided to look in a completely different direction. The NLS variant and its solitons can in principle be materialized experimentally with optical devices, suitably etched, see e.g. \cite{Segev review}. Thus the dynamics of NLS with such solitons may offer new and potentially useful coherent states for optical devices. Furthermore, we believe the NLS solitons to have special properties, in particular asymptotic stability as opposed to the asymptotic metastability of the NLKG solitons conjectured in \cite{CFW}.


We will be following a procedure for the proof of asymptotic stability which has become standard within the study of nonlinear PDE \cite{Avy NLS}. Crucial aspects of the theory and associated results were established by Buslaev and Perelman \cite{important results 1}, Buslaev and Sulem \cite{important results 2}, and Gang and Sigal \cite{important results 3}. Important elements of these methods are the dispersive estimates. Various such estimates have been found in the context of 1D lattice systems, for example see the work of A.I. Komech, E.A. Kopylova, and M. Kunze \cite{important results 4} and of I. Egorova, E. Kopylova, G. Teschl \cite{1D lattice decay estimates}, as well as the continuum 2D problem to which our system bears many resemblances, see e.g. the work of E. A. Kopylova and A.I. Komech \cite{2D}. Extensive results have been found on the asymptotic stability on solitons of 1D nonlinear lattice Schr\"odinger equations by F. Palmero et al. \cite{important results 5}, P.G. Kevrekidis, D.E. Pelinovsky, and A. Stefanov \cite{important results 6}, as well as S. Cuccagna and M. Tarulli \cite{CucTar}. Typically the literature on 1D lattice NLS systems focuses on cases where the free linear Schr\"odinger operator is given by the negative of the standard 1D discrete Laplacian. Our work is on a different free linear Schr\"odinger operator, $L_0$ defined above, which has some distinguishing properties. Important aspects of the application of these models to optical nonlinear waveguide arrays has been established by H.S. Eisenberg et al. \cite{important results 7}.


This work is the first of a series of papers (this one, \cite{paper 02} and \cite{paper 03}) devoted to the construction, scattering, and asymptotic stability of radial noncommutative solitons with two noncommuting spatial coordinates. We have chosen to restrict our study to these solutions for a number of reasons: it builds upon the observations and results of \cite{CFW}; the radial cases allow one to work with effective 1D lattices and thereby standard Jacobi operators; for two noncommuting spatial coordinates the free radial system is equivalent to a known Jacobi operator spectral problem; the method proposed is by far the most illustrative for the given restrictions. The three papers are devoted to separate aspects of the problem in order of necessity. The organization of this work is as follows.

In this paper we focus on a key estimate that is needed for scattering and stability, namely the decay in time of solutions of relevant Schr\"odinger operators. Fortunately, for boundary perturbed operators, we find it is integrable, given by $t^{-1}\log^{-2}t$. The proof of this result is rather direct, and employs the generating functions of the corresponding generalized eigenfunctions, to explicitly represent and estimate the resolvent of the hamiltonian at all energies. We follow the general approach established by Jensen and Kato \cite{JenKat} and extended by Murata \cite{Murata} whereby time decay follows largely from the behavior of the resolvent near the threshold. From this one can see that for the chosen weights the estimate we find is optimal and should be optimal in general due to the elimination of the threshold resonance by boundary perturbations, by the generality of the method. We also conclude the absence of positive eigenvalues and singular continuous spectrum.

Previous results for the scattering theory of the associated noncommutative waves and solitons were found by Durhuus and Gayral \cite{noncommutative scattering}. In particular they find local decay estimates for the associated noncommutative NLS. They consider general noncommutative estimates for all for all even dimensions of pairwise noncommuting spaces. We consider radial solutions on 2D noncommutative space by alternative methods and find local decay for both the free Schr\"odinger operator as well as a class of rank one perturbations thereof. Our decay estimates are an improvement on those of \cite{noncommutative scattering} for this restricted class of solutions. An important element of this analysis is the study of the spectral properties of the free and boundary-perturbed Schr\"odinger operator. The boundary-perturbation is crucial to the work as it not only eliminates the threshold resonance of the free operator (thereby improving the time decay) as well as allows one to approximate and control solitons that are large only at the boundary via linear operators. We extend the linear analysis of Chen, Fr\"ohlich, and Walcher \cite{CFW} and reproduce some of their results with alternative techniques.

In \cite{paper 02} we address the construction and properties of a family of ground state solitons. These stationary states satisfy a nonlinear eigenvalue equation, are positive, monotonically decaying and sharply peaked for large spectral parameter. The proof of this result follows directly from our spectral results in this paper by iteration for small data and root finding for large data. The existence and many properties of solutions for a similar nonlinear eigenvalue equation were found by Durhuus, Jonssen, and Nest \cite{DJN 1}\cite{DJN 2}. We utilize a simple power law nonlinearity for which their existence proofs do not apply. We additionally find estimates for the peak height, spatial decay rate, norm bounds, and parameter dependence.

In \cite{paper 03} we focus on deriving a decay rate estimate for the Hamiltonian which results from linearizing the original NLS around the soliton constructed in \cite{paper 02}. We determine the full spectrum of this operator, which is the union of a multiplicity 2 null eigenvalue and a real absolutely continuous spectrum. This establishes a well-defined set of modulation equations \cite{Avy NLS} and points toward the asymptotic stability of the soliton.

In the conclusion of \cite{paper 03} we describe how the results can be applied to prove stability of the soliton we constructed in \cite{paper 02}. The issue of asymptotic stability of NLS solitons has been sufficiently well-studied in such a broad context that the proof thereof is often considered as following straightforwardly from the appropriate spectral and decay estimates, of the kind found in \cite{paper 03}. We sketch how the theory of modulation equations established by Soffer and Weinstein \cite{Avy NLS} can be used to prove asymptotic stability. Chen, Fr\"ohlich, and Walcher \cite{CFW} conjectured that in the NLKG case the corresponding solitons are unstable but with exponential long decay: the so-called metastability property, see \cite{Avy NLKG} . There is a great deal of evidence to suggest that this is in fact the case but a proof has yet to be provided. This will be the subject of future work.

\section{Notation}

Let $\mathbb{Z}_+$ and $\mathbb{R}_+$ respectively be the nonnegative integers and nonnegative reals and $\mathscr{H} = \ell^2(\mathbb{Z}_+,\mathbb{C})$ the Hilbert space of square integrable complex functions, e.g. $v: \mathbb{Z}_+ \ni x \mapsto v(x) \in \mathbb{C}$, on the 1D half-lattice with inner product $( \cdot , \cdot )$, which is conjugate-linear in the first argument and linear in the second argument, and the associated norm $||\cdot||$, where $||v|| = (v,v)^{1/2}$, $\forall v\in\mathscr{H}$. Where the distinction is clear from context $||\cdot|| \equiv ||\cdot||_{\mathrm{op}}$ will also represent the norm for operators on $\mathscr{H}$ given by $||A||_{\mathrm{op}} = \sup_{v \in \mathscr{H}}||v||^{-1}||Av||$, for all bounded $A$ on $\mathscr{H}$. Denote the lattice $\ell^1$ norm by $||\cdot||_1$ where $||v||_1 = \sum_{x=0}^\infty|v(x)|$, $\forall v \in \ell^1(\mathbb{Z}_+,\mathbb{C})$.

We denote by $\otimes$ the tensor product and by $z \mapsto \overline{z}$ complex conjugation for all $z \in \mathbb{C}$. We write $\mathscr{H}^*$ for the space of linear functionals on $\mathscr{H}$: the dual space of $\mathscr{H}$. For every $v \in \mathscr{H}$ one has that $v^* \in \mathscr{H}^*$ is its dual satisfying $v^{*}(w) = (v,w)$ for all $v,w \in \mathscr{H}$.  For every operator $A$ on $\mathscr{H}$ we take $\mathcal{D}(A)$ as standing for the domain of $A$. For each operator $A$ on $\mathscr{H}$ define $A^*$ on $\mathscr{H}^*$ to be its dual and $A^\dag$ on $\mathscr{H}$ its adjoint such that $v^*(Aw) = A^*v^*(w) = (A^\dag v, w)$ for all $v \in \mathcal{D}(A^\dag)$ and all $w \in \mathcal{D}(A)$. Let $\{\chi_{x}\}_{x=0}^{\infty}$ be the orthonormal set of vectors such that $\chi_{x}(x) = 1$ and $\chi_{x_{1}}(x_{2}) = 0$ for all $x_{2} \ne x_{1}$. We write $P_{x} = \chi_{x} \otimes \chi^{*}_{x}$ for the orthogonal projection onto the space spanned by $\chi_{x}$.

We define $\mathscr{T}$ to be the topological vector space of all complex sequences on $\mathbb{Z}_+$ endowed with topology of pointwise convergence, $\mathcal{B}(\mathscr{H})$ to be the space of bounded linear operators on $\mathscr{H}$, and $\mathcal{L}(\mathscr{T})$ to be the space of linear operators on $\mathscr{T}$, endowed with the pointwise topology induced by that of $\mathscr{T}$. When an operator $A$ on $\mathscr{H}$ can be given by an explicit formula through $A(x_{1},x_{2}) = (\chi_{x_{1}},A\chi_{x_{2}}) < \infty$ for all $x_{1},x_{2} \in \mathbb{Z}_{+}$ one may make the natural inclusion of $A$ into $\mathcal{L}(\mathscr{T})$, the image of which will also be denoted by $A$. We consider $\mathscr{T}$ to be endowed with pointwise multiplication, i.e. the product $uv$ is specified by $(uv)(x)=u(x)v(x)$ for all $u,v \in \mathscr{T}$.

We represent the \emph{spectrum} of each $A$ on $\mathscr{H}$ by $\sigma(A)$. We term each element $\lambda \in \sigma(A)$ a \emph{spectral value}. We write $\sigma_{\mathrm{d}}(A)$ for the \emph{discrete spectrum}, $\sigma_{\mathrm{e}}(A)$ for the \emph{essential spectrum}, $\sigma_{\mathrm{p}}(A)$ for the \emph{point spectrum}, $\sigma_{\mathrm{ac}}(A)$ for the \emph{absolutely continuous spectrum}, and $\sigma_{\mathrm{sc}}(A)$ for the \emph{singularly continuous spectrum}. Should an operator $A$ satisfy the spectral theorem there exist scalar measures $\{\mu_{k}\}_{k=1}^{n}$ on $\sigma(A)$ which furnish the associated spectral representation of $\mathscr{H}$ for $A$ such that the action of $A$ is given by multiplication by $\lambda \in \sigma(A)$ on $\oplus_{k=1}^{n}L^{2}(\sigma(A),\mathrm{d}\mu_{k})$. If $\mathscr{H} = \oplus_{k=1}^{n}L^{2}(\sigma(A),\mathrm{d}\mu_{k})$ we term $n$ the \emph{generalized multiplicity} of $A$. For an operator of arbitrary generalized multiplicity we will write $\mu^{A}$ for the associated operator valued measure, such that $A = \int_{\sigma(A)} \lambda\ \mathrm{d}\mu^{A}_{\lambda}$. For each operator $A$ that satisfies the spectral theorem, its spectral (Riesz) projections will be written as $P^{A}_{\mathrm{d}}$ and the like for each of the distinguished subsets of the spectral decomposition of $A$. Define $R^A_\cdot: \rho(A) \to \mathcal{B}(\mathscr{H})$, the resolvent of $A$, to be specified by $R^A_z := (A - z)^{-1}$, where $\rho(A) := \mathbb{C} \setminus \sigma(A)$ is the resolvent set of $A$ and where by abuse of notation $zI \equiv z \in \mathcal{B}(\mathscr{H})$ here.

Allow an \emph{eigenvector} of $A$ to be a vector $v \in \mathscr{H}$ for which $Av = \lambda v$ for some $\lambda \in \mathbb{C}$. Should $A$ admit inclusion into $\mathcal{L}(\mathscr{T})$, we define a \emph{generalized eigenvector} of $A$ be a vector $\phi \in \mathscr{T} \setminus \mathscr{H}$ which satisfies $A\phi = \lambda \phi$ for some $\lambda \in \mathbb{C}$ such that $\phi(x)$ is polynomially bounded, which is to say that there exists a $p \ge 0$ such that $\lim_{x \nearrow \infty}(x+1)^{-p}\phi(x) = 0$. We define a \emph{spectral vector} of $A$ to be a vector which is either an eigenvector or generalized eigenvector of $A$. We define the subspace of spectral vectors associated to the set $\Sigma \subseteq \sigma(A)$ to be the \emph{spectral space over $\Sigma$}.

We write $\partial_z \equiv \frac{\partial}{\partial z}$ and $\mathrm{d}_z \equiv \frac{\mathrm{d}}{\mathrm{d} z}$ respectively for formal partial and total derivative operators with respect to a parameter $z \in \mathbb{R}, \mathbb{C}$.

\section{Results}

\begin{defn}
Define $L_0$ to be the operator on $\mathscr{H}$ with action
\begin{align}
	L_0v(x) = \left\{
	\begin{array}{cc}
		- (x+1)v(x+1) + (2x + 1)v(x) - x v(x-1) &,\quad x > 0 \\
		- v(1) + v(0) &,\quad x = 0 .
	\end{array} \right.
\end{align}
and domain $\mathcal{D}(L_0) := \{ v \in \mathscr{H}\ |\ || Mv || < \infty \}$, where $M$ is the multiplication operator with action $Mv(x) = xv(x)$ $\forall v \in \mathscr{T}$.
\end{defn}

Consider the linear Schr\"odinger equation
\begin{align}\label{SE}
	i\partial_tu = L_0u +Vu
\end{align}
where $u : \mathbb{R}_t \times \mathbb{Z}_+ \to \mathbb{C}$ and $V$ is a potential (energy) multiplication operator on $\mathscr{H}$. To find solutions to Equation \eqref{SE} it is sufficient to analyze the spectral measure of $L_{0} + V$. The regularity and boundedness properties of $V$ are crucial for the analysis of solutions and for general results such must be given in advance. In our work any potential introduced will be given explicitly so all relevant properties will be given by its representation and domain of definition.

\begin{snprop}\label{snprop01}
The operator $L_0$ has the following properties.
\begin{enumerate}
	\item $L_0$ is essentially self-adjoint.
	\item $L_{0}$ has generalized multiplicity 1.
	\item The spectrum of $L_0$ is absolutely continuous, $\sigma(L_0) = \sigma_{\mathrm{ac}}(L_0) = [0,\infty)$, and its generalized eigenfunctions are the Laguerre polynomials $\phi_{\lambda}(x) \equiv \phi^{L_0}_{\lambda}(x) = \sum_{k=0}^x \frac{(-\lambda)^k}{k!}\binom{x}{k}$ for choice of normalization $\phi_{\lambda}(0) = 1$.
\end{enumerate}
\end{snprop}

\noindent Chen, Fr\"ohlich, and Walcher determined the above properties for $L_{0}$ in \cite{CFW} via methods which are different from ours. We refer the reader to the presentation of the spectral projection $\delta^{L_0}_\lambda$ in Definition 5.3.

\begin{defn}
Let $w_{\lambda} \equiv w^{L_0}_{\lambda} := (\chi_{0},\delta^{L_{0}}_{\lambda} \chi_{0})$ be termed the \emph{spectral integral weight}, $\psi_z \equiv \psi^{L_0}_z := R^{L_{0}}_z\chi_{0}$ the \emph{resolvent vector}, $\xi_z \equiv \xi^{L_0}_z := \psi_z - \psi_z(0)\phi_z$ the \emph{auxilliary resolvent vector}, and $f_{z} \equiv f^{L_0}_z := (\chi_{0},R^{L_{0}}_{z}\chi_{0})$ the \emph{resolvent function of $L_{0}$} for all $\lambda \in \sigma(L_{0})$ for all $z \in \rho(L_{0})$.
\end{defn}

Since $\phi_{\lambda}(x)$ is a polynomial of degree $x$ in $\lambda$, one has that the analytic continuation $\phi_{z} \in \mathscr{T}$, $z \in \mathbb{C}$, exists. The above permits the useful representation $\psi_{z} = f_{z}\phi_{z}+\xi_{z}$. $\phi_{z}$ and $\psi_{z}$ are connected to the Stieltjes transform theory of orthogonal polynomials. In that context $\phi_{z}$ is termed a \emph{primary polynomial} and $\psi_{z}$ the corresponding \emph{secondary polynomial} \cite{moment}.

\begin{defn}
Define $L$ to be the operator on $\mathscr{H}$ with domain $\mathcal{D}(L) = \mathcal{D}(L_0)$ and specified by $L := L_0 - q P_0$ where $q > 0$ is a fixed constant and $P_0 := \chi_0 \otimes \chi_0^*$. Let $\psi^{L}_{z} := R^{L}_{z}\chi_{0}$ be the resolvent vector and $f^{L}_{z} = (\chi_{0},R^{L}_{z}\chi_{0})$ the resolvent function of $L$ for all $z \in \rho(L)$.
\end{defn}

The addition of the attractive boundary potential enhances the time decay. It also allows one to approximate and control solutions with data that is only large at the boundary in applications to nonlinear problems.

\begin{snthm}\label{snthm01}
Let $\phi^L_\lambda$, $\lambda \in \sigma(L)$, denote spectral vectors of $L$ chosen to satisfy the normalization condition $(\chi_{0},\phi^{L}_{\lambda}) = \phi^L_\lambda(0) = 1$, $\forall \lambda \in \sigma(L)$. $L$ has the following properties.
	\begin{enumerate}
		\item $\sigma_{\mathrm{d}}(L) = \sigma_{\mathrm{p}}(L) = \{ \lambda_0 \}$, where $\lambda_0 < 0$ uniquely satisfies $1 = q \psi_{\lambda_0}(0)$ and the unique eigenfunction over $\lambda_0$ is $\psi^L_{\lambda_0} = q\psi_{\lambda_0}$.
		\item $\sigma_{\mathrm{e}}(L) = \sigma_{\mathrm{ac}}(L) = \sigma(L_{0}) = [0,\infty)$ and has generalized multiplicity 1.
		\item $\mathrm{d}\mu^L(\lambda) = w^L_{\lambda}\phi^L_{\lambda} \otimes \phi^{L,*}_{\lambda} \ \mathrm{d}\lambda$, where $w^L_{\lambda} = \{ [1+ q e^{-\lambda} \mathrm{Ei}(\lambda)]^2 + [\pi q e^{-\lambda}]^2 \}^{-1} e^{-\lambda}$, $\mathrm{d}\lambda$ is the Lebesgue measure on $[0,\infty)$, and $\mathrm{Ei}(\lambda) := \int_{-\lambda}^\infty \!\mathrm{d}u\ u^{-1}e^{-u}$, $\lambda>0$, is the exponential integral. The generalized eigenfunctions of $L$ are given by $\phi^L_\lambda = \phi_{\lambda} + q \xi_{\lambda}$, $\lambda \in \sigma_{\mathrm{ac}}(L)$.
	\end{enumerate}
\end{snthm}

We are ultimately interested in studying the solutions of a nonlinear equation so it is important to acquire decay estimates for dispersive ``scattering states''.

\begin{defn}
Let $W_{\kappa,\tau}$ be the multiplication operator weight specified by \\ $W_{\kappa,\tau}v(x) = (x + \kappa)^{\tau}v(x)$, $\forall v \in \mathscr{T}$, where $0 < \kappa \in \mathbb{R}$, $ \tau \in \mathbb{R}$.
\end{defn}

\begin{snthm}\label{snthm02}
For all $-3 \ge \tau \in \mathbb{R}$, $t > 0$, $v \in \ell^{1}$, there exists a constant $c > 0$ and a $1 < \kappa \in \mathbb{R}$ such that
\begin{align}
	|| W_{\kappa,\tau}e^{-itL_{0}}W_{\kappa,\tau}v ||_{\infty} < ct^{-1}|| v ||_{1}
\end{align}
\end{snthm}

\begin{snthm}\label{snthm03}
For all $-3 \ge \tau \in \mathbb{R}$, $v \in \ell^{1}$, there exists a $1 < \kappa \in \mathbb{R}$ such that
\begin{align}
	|| W_{\kappa,\tau}e^{-itL}P^{L}_{\mathrm{e}}W_{\kappa,\tau}v ||_{\infty} = \mathcal{O}(t^{-1}\log^{-2}t),\quad t\nearrow\infty
\end{align}
\end{snthm}

Our proof of these estimates will rely heavily on the generating functions of the generalized eigenvalues. This approach draws upon known properties of certain special functions. Hereby the problem of sequences on a lattice will be transformed into a problem of analytic functions in the complex plane.

\section{Spectral Properties of $L_{0}$}

\begin{lem}
Any vector, $v$, the set of whose components, $\{v(x)\}_{x=0}^{\infty}$, have finitely many nonzero elements is a semi-analytic vector for $L_{0}$, which is to say that $||L^{k}_{0}v|| \le c_{v} (2k)!$ where $c_{v}$ depends on $v$ alone.
\end{lem}

\begin{proof}
Define $x_{v} := \sup_{x}\{ x : v(x) \neq 0 \}$.
\begin{align}
	||L_{0}v||^{2}_{2} &= \sum_{x=0}^{\infty} |-(x+1)v(x+1) + (2x+1)v(x) - xv(x-1)|^{2} \\
		&\le \sum_{x=0}^{\infty} [(x+1)|v(x+1)| + (2x+1)|v(x)| + x|v(x-1)|]^{2} \\
		&\le \sum_{x=0}^{\infty}\{ [(x_{v}-1)+1] |v(x+1)| + (2x_{v}+1)|v(x)| + (x_{v}+1)|v(x-1)| \}^{2} \\
		&= \sum_{x=0}^{\infty}\{ [ x_{v} |v(x+1)| ]^{2} + 2x_{v}(2x_{v}+1) |v(x+1)| |v(x)| \\
			&\quad\quad + 2x_{v}(x_{v}+1)|v(x)| |v(x-1)| + [(2x_{v}+1)|v(x)|]^{2} \\
			&\quad\quad + 2(2x_{v}+1)(x_{v}+1)|v(x)| |v(x-1)| + [(x_{v}+1)|v(x-1)|]^{2} \} \\
		&\le \sum_{x=0}^{\infty}\{ [ x_{v} |v(x)| ]^{2} + 2x_{v}(2x_{v}+1) |v(x+1)| |v(x)| \\
			&\quad\quad + 2x_{v}(x_{v}+1)|v(x+1)| |v(x)| + [(2x_{v}+1)|v(x)|]^{2} \\
			&\quad\quad + 2(2x_{v}+1)(x_{v}+1)|v(x+1)| |v(x)| + [(x_{v}+1)|v(x)|]^{2} \} \\
		&\le 16(x_{v}+1)^{2}||v||^{2}_{1}
\end{align}
\begin{align}
	||L_{0}v||_{1} &= \sum_{x=0}^{\infty} |-(x+1)v(x+1) + (2x+1)v(x) - xv(x-1)| \\
		&\le \sum_{x=0}^{\infty} |(x+1)v(x+1) + (2x+1)v(x) + xv(x-1)| \\
		&\le \sum_{x=0}^{\infty} [(x+1)|v(x+1)| + (2x+1)|v(x)| + x|v(x-1)|]
\end{align}
\begin{align}
		&\le \sum_{x=0}^{\infty} \{[(x_{v}-1)+1]|v(x+1)| + (2x_{v}+1)|v(x)| + (x_{v} + 1)|v(x-1)|\} \\
		&\le \sum_{x=0}^{\infty} [ x_{v}|v(x+1)| + (2x_{v}+1)|v(x)| + (x_{v} + 1)|v(x-1)| ] \\
		&\le 4(x_{v}+1) ||v||_{1} .
\end{align}
Let $a_{v} := 4(x_{v}+1)$. We have then that $||L_{0}v||^{2}_{1}, ||L_{0}v||^{2}_{2} \le a_{v} ||v||_{1}$. One may observe that $x_{L_{0}v} = x_{v} + 1 \Rightarrow a_{L_{0}v} = a_{v} + 4$. One then has that
\begin{align}
	||L^{k}_{0}v||_{2} &\le a_{L^{k-1}_{0}v}||L^{k-1}_{0}v||_{1} \le a_{L^{k-1}_{0}v}a_{L^{k-2}_{0}v}||L^{k-2}_{0}v||_{1} \le \ldots \\
		&\le \prod_{j = 1}^{k} a_{L^{k-j}_{0}v} ||v||_{1} = 4^{k} \prod_{j = 1}^{k} (j + x_{v}) ||v||_{1} = 4^{k}(x_{v}!)^{-1}(k+x_{v})! ||v||_{1}.
\end{align}
Since $4^{k}(x_{v}!)^{-1}(k+x_{v})! ||v||_{1}$ is monotonically increasing in $k \in \mathbb{Z}_{+}$ we may, without loss of generality, take that $k > x_{v}$ and $k > 4$ to bound this expression. One may show through the monotonicity in $k$ of $\binom{x}{k}$ for $0 \le k \le \lfloor x/2 \rfloor$ that $\binom{x}{k} \le (\lfloor x/2 \rfloor !)^{-2}(x!)$, where $\lfloor a \rfloor = \sup_{a \ge n \in \mathbb{Z}} n$, for all $a \in \mathbb{R}$, is the floor function. One then has
\begin{align}
	4^{k}(x_{v}!)^{-1}(k+x_{v})! ||v||_{1} &= 4^{k}k!\binom{k + x_{v}}{k} ||v||_{1} \le 4^{k}k!\binom{2k}{k} ||v||_{1} \\
		&= 4^{k}(k!)^{-1} (2k)! ||v||_{1} \le 4^{4}(4!)^{-1}(2k)! ||v||_{1}\\
		&< c_{v}(2k)!,
\end{align}
where $c_{v} = 11||v||_{1}$.
\end{proof}

\begin{proof}[Proof of Proposition \ref{snprop01} Part (1)]
The set of vectors, $v$, with finitely many nonzero components is dense in $\mathscr{H}$. This dense set is semi-analytic for $L_{0}$. By Semi-Analytic Vector Theorem, see e.g. the appendix or \cite{semi-analytic}, it is therefore the case that $L_{0}$ is essentially self-adjoint.
\end{proof}

\begin{proof}[Proof of Proposition \ref{snprop01} Part (2)]
One has that $L_0v(x) = \lambda v(x)$, $v \in \mathscr{T}$, specifies a countable family of coupled elementary algebraic equations. A unique solution may be found for each $\lambda$ by specifying $v_0$ and solving inductively in increasing $x \in \mathbb{Z}_+$. A choice of normalization will fix $v(0)$. Therefore each solution is, up to normalization, uniquely specified by $\lambda$.
\end{proof}

\begin{proof}[Proof of Proposition \ref{snprop01} Part (3)]
$L_0$ is an essentially self-adjoint, second order, finite difference operator or \emph{Jacobi operator}. It is well known that the theory of Jacobi operators is intimately connected with that of orthogonal polynomials. In particular spectral equations for operators extended to formal sequence spaces may be viewed as recursion formulas for families, indexed by lattice site, of orthogonal polynomials defined on the spectrum of the operator in question, see e.g. \cite{Jacobi}. For $L_0 \in \mathcal{L}(\mathscr{T})$ it is the case that $L_0v(x) = \lambda v(x)$ takes the form of the recursion formula for the Laguerre polynomials. By part (2) of the proposition these solutions are unique.
\end{proof}

The Laguerre polynomials, $\phi_\lambda(x)$, have known completeness and orthogonality relations whose roles will be reversed here:
\begin{align}
	\delta_{x_{1},x_{2}} = \int_0^\infty \mathrm{d}\lambda\ e^{-\lambda} \phi_\lambda(x_{1})\phi_\lambda(x_{2}),\quad \delta(\lambda_{1} - \lambda_{2}) = e^{-(\lambda_{1}+\lambda_{2})/2}\sum_{x=0}^\infty \phi_{\lambda_{1}}(x)\phi_{\lambda_{2}}(x) ,
\end{align}
where here $\delta(\cdot)$ is Dirac's delta distribution supported on $\sigma(L_0)$. The RHS of these equations converge in the distributional sense respectively on $\ell^2(\mathbb{Z}_+)$ and $L^2(\mathbb{R}_+)$. The former equation expresses components of the spectral measure of $L_0$ and in particular $w_{\lambda} = e^{-\lambda}$.

\section{Spectral Properties of $L$}

\begin{defn}
	Let $\mathcal{T} \in \mathcal{L}(\mathscr{T})$ be the \emph{binomial transform}, see e.g. \cite{binomial}, defined by
		\begin{align}
			\mathcal{T}v(k) = \sum_{x=0}^\infty \mathcal{T}(k,x)v(x) = \sum_{x=0}^\infty (-1)^x\binom{k}{x} v(x) ,\quad \forall v \in \mathscr{T}.
		\end{align}
\end{defn}

$\mathcal{T}$ is involutive in the sense that $\mathcal{T}^2 = I$. One has that $\mathcal{T}v(0) = v(0)$ and the useful representation $\chi_0(x) = \sum_{k=0}^\infty (-1)^k\binom{x}{k}$. We take the conventions that $x! = \binom{x}{k} = \sum_{x=0}^k v(x) = 0$ for $k,x < 0$ and $k < x$ for all $v \in \mathscr{T}$.

\begin{lem}
It is the case that that
\begin{align}
	\mathcal{T}L_0v(k) = (k+1)\mathcal{T}v(k+1) ,\quad \forall v \in \mathscr{T} .
\end{align}
\end{lem}

\begin{proof}
One may write $\mathcal{T}v(k) = (u^\mathcal{T}_k, v)$ where $u^\mathcal{T}_k(x) = (-1)^x\binom{k}{x}$. By the symmetry of $L_0$ one has $\mathcal{T}L_0v(k) = (u^\mathcal{T}_k, L_0v) = (L_0u^\mathcal{T}_k, v)$ and it is therefore sufficient to analyze $L_0u^\mathcal{T}_k$ alone. We will utilize the formulas $\binom{k-1}{x-1} = \binom{k}{x} - \binom{k-1}{x} = \frac{x}{k}\binom{k}{x}$. \\
\underline{For $x=0$:}
\begin{align}
L_0u^\mathcal{T}_k(x) = -(-1)^{x+1}\binom{k}{x+1} + (-1)^x\binom{k}{x} = (k+1) = (k+1)u^\mathcal{T}_{k+1}(0).
\end{align}
\underline{For $x>0$:}
\begin{align}
L_0u^\mathcal{T}_k(x) &= -(x+1)(-1)^{x+1}\binom{k}{x+1} + (2x+1)(-1)^x\binom{k}{x} - x(-1)^{x-1}\binom{k}{x-1} \\
&= (-1)^x\left[ (x+1)\binom{k}{x+1} + (2x + 1)\binom{x}{k} + x\binom{k}{x-1} \right] \\
&= (-1)^x\left[ (x+1)\binom{k+1}{x+1} - (x+1)\binom{k}{x} + (2x + 1)\binom{x}{k}\right. \\
&\quad\quad \left. + x\binom{k+1}{x} - x\binom{k}{x} \right] \\
&= (-1)^x\left[ (x+1)\binom{k+1}{x+1} + x\binom{k+1}{x} \right] = (k+1)(-1)^x\binom{k+1}{x} .
\end{align}
\end{proof}
One may then use the binomial transform to arrive at the standard power series representation of the Laguerre polynomials. Consider
\begin{align}
L_0 \phi_{\lambda}(x) = \lambda \phi_{\lambda}(x) \quad\Rightarrow\quad (k+1)\mathcal{T}\phi_{\lambda}(k+1) = \lambda\mathcal{T}\phi_{\lambda}(k).
\end{align}
Choosing $\mathcal{T}\phi_{\lambda}(0) = \phi_{\lambda}(0) =1$ one has by induction that $\mathcal{T}\phi_{\lambda}(k) = \frac{\lambda^k}{k!}$. One may then apply the binomial transform again to arrive at
\begin{align}
	\phi_{\lambda}(x) = \sum_{k=0}^x (-1)^k\binom{x}{k}\frac{(\lambda)^k}{k!} .
\end{align}

\begin{lem}
One has the representation
\begin{align}
	\psi_{z}(x) = e^{-z} \sum_{k=0}^x (-1)^k\binom{x}{k}E_{k+1}(-z),
\end{align}
where
\begin{align}
	E_p(z) := z^{p-1}\int_z^\infty \mathrm{d}t\ e^{-t} t^{-p},\qquad p\in \mathbb{C}, z\in \mathbb{C}\setminus(-\infty,0]
\end{align}
are the \emph{generalized exponential integrals} for which we take the principal branch with standard branch cut $\Sigma = (-\infty, 0]$.
\end{lem}

\begin{proof}
Consider the $(L_0 - z)\psi_{z} = \chi_0$, where $\psi_{z},\chi_{0} \in \mathscr{T}$, $L_{0} \in \mathcal{L}(\mathscr{T})$. By binomial transform of this equation one finds
\begin{align}
	(k+1)\mathcal{T}\psi_{z}(k+1) = z \mathcal{T}\psi_{z}(k) + 1.
\end{align}
$\mathcal{T}\psi_{z}(k) = e^{-z} E_{k+1}(-z)$ satisfies this recursion formula.
\end{proof}
The generalized exponential integrals have many other integral representations however most are defined only on a restricted set of $p, z$.
\begin{defn}
For any single-valued or multi-valued function $f: \mathbb{C} \to \mathbb{C}$, an element of a set of linear functionals on some suitable Banach space with norm given through integration over $\lambda$, and with poles, branch points, and branch cuts found in the subset $\Sigma \subseteq \mathbb{R}$ let $\mathcal{PV}f : \Sigma \to \mathbb{C}$ be the \emph{principal value of $f$} defined by the weak limit
\begin{align}
	\mathcal{PV}f(\lambda) :=& \frac{1}{2} \wlim_{\epsilon \searrow 0} \left[  f(\lambda + i\epsilon) + f(\lambda - i\epsilon) \right] , \quad \lambda \in \Sigma,
\end{align}
which converges in the distributional sense. We analogously define the \emph{$\delta$-part} of $f$ to be
\begin{align}
	\delta f(\lambda) :=& {1 \over 2\pi i} \wlim_{\epsilon \searrow 0} \left[  f(\lambda + i\epsilon) - f(\lambda - i\epsilon) \right] , \quad \lambda \in \Sigma.
\end{align}
\end{defn}

We have kept vague the specification of the sense in which the above definitions converge weakly for the purposes of generality. The details of such convergence in our work will be clear from context. One may extend the domain of $\mathcal{PV}f$ to the complex plane and produce a single valued function, which we will also denote $f$, through
\begin{align}
	\mathcal{PV}f(z) := \left\{
	\begin{array}{cr}
		 f(z),&\ z \in \mathbb{C} \setminus \Sigma \\
		\mathcal{PV}f(z),&\ z \in \Sigma .
	\end{array} \right. .
\end{align}
One may observe that the analogous extension of $\delta f(\lambda)$ vanishes away from $\Sigma \subseteq \mathbb{R}$. This prescription extends to weak limits in $z \in \mathbb{C}$ of complex sequences $v_{z} \in \mathscr{T}$ whose components depend upon $z$.

The generalized exponential integrals have the convergent series expansion \cite{En}
\begin{align}
	E_{n+1}(z) &= -\frac{(-z)^n}{n!} \log(z) + \frac{e^{-z}}{n!}\sum_{k=1}^n(-z)^{k-1}(n-k)!\\
	&\qquad+ \frac{e^{-z}(-z)^n}{n!}\sum_{k=0}^\infty \frac{z^k}{k!}\digamma(k+1) ,
\end{align}
where $\digamma(z) := \mathrm{d}_z \log \Gamma(z)$ is the digamma function. One may therefore observe that
\begin{align}
	\wlim_{\epsilon \searrow 0} E_{n+1}(-z \pm i \epsilon) = \mathcal{PV}E_{n+1}(-z) \mp i\pi \frac{(x)^n}{n!} ,\quad z > 0,
\end{align}
where for the sake of generality the limit is weak with respect to $L^{2}([a,\infty),\mathbb{C})$, $a > 0$.
One may write $\mathcal{PV}E_1(-z) = -\mathrm{Ei}(z)$ where
\begin{align}
	\mathrm{Ei}(x) := - \int_{-z}^\infty \!\mathrm{d}u\ u^{-1}e^{-u},\qquad z>0
\end{align}
is \emph{the exponential integral}.

\begin{proof}[Proof of Theorem \ref{snthm01} Part (1)]

Let $u \in \mathscr{H}$ satisfy $u(0) \neq 0$, then
\begin{align}
	Lu = z u \quad \Rightarrow \quad 1 = qf^{L_{0}}_{z}.
\end{align}
$L$ will then have as many eigenvalues as $qf^{L_{0}}_{z} - 1$ has zeroes. The corresponding eigenfunctions are given by
\begin{align}
	Lu = \lambda u \quad \Rightarrow \quad u = q u(0)R^{L_{0}}_\lambda \chi_{0} .
\end{align}
Here $f_{z} = e^{-z}E_1(-z)$ so eigenvalues are be given by zeros of $qe^{-z}E_1(-z) - 1$. $L$ is essentially self-adjoint so $\sigma(L) \subseteq \mathbb{R}$. Analytic continuation of $E_1(-z)$ to $z > 0$ from above or below will result in the sum of a real function and an imaginary constant
\begin{align}
	\lim_{\epsilon \searrow 0} E_1(-x \pm i \epsilon) = - \mathrm{Ei}(x) \mp i\pi ,\quad x > 0
\end{align}
so there can be no positive eigenvalues. $qe^{-z}E_1(-z)$ diverges for $z \to 0$ so $z = 0$ cannot be an eigenvalue. All eigenvalues must be negative. Let $z = - a < 0$. It is the case that $e^aE_1(a)$ is monotonically decreasing for increasing $a \in (0,\infty]$
\begin{align}
	\mathrm{d}_a[e^aE_1(a)] = -\int_0^\infty \mathrm{d}x\ e^{-x} (x+a)^{-2} < 0,
\end{align}
where we have used an alternative integral representation of $E_1(z)$ and manifest dominated convergence of the integral to pass the derivative through the integral. Furthermore since
\begin{align}
	\lim_{a \searrow 0} \int_0^\infty \mathrm{d}t\ e^{-t} (t+a)^{-1} &= \infty,\quad \lim_{a \nearrow \infty} \int_0^\infty \mathrm{d}t\ e^{-t} (t+a)^{-1} = 0,
\end{align}
it follows that $e^aE_1(a)$ takes each on the interval $[0,\infty)$ exactly once, where we have used manifest uniform convergence of the integrand to pass the limit through the integral. Therefore $qe^{-z}E_1(-z) - 1$ has exactly one root for each fixed $q>0$.
\end{proof}

\begin{proof}[Proof of Theorem \ref{snthm01} Part (2)]
By the argument of the Proof of Theorem \ref{snthm01} Part (1), there can be no embedded eigenvalues. By Weyl's criterion the perturbation of $L_{0} \mapsto L$ leaves the essential spectrum unchanged. The argument for the proof of $\sigma(L_{0}) = \sigma_{\mathrm{ac}}(L_{0})$ follows without change for the spectrum of $L$.
\end{proof}

\begin{defn}
Let $A$ be an operator on $\mathscr{H}$ which is self-adjoint on its domain $\mathcal{D}(A)$ and $\lambda$ an element of the discrete spectrum of $A$. Define $\mathcal{PV}^A_\lambda \equiv \mathcal{PV}(A - \lambda)^{-1}$, $\lambda \in \sigma(A)$ to be the \emph{principal value of the resolvent of $A$} given by the strong limit
\begin{align}
	\mathcal{PV}^A_\lambda :=& \frac{1}{2}\slim_{\epsilon \searrow 0} \left[  R^A_{\lambda + i\epsilon} + R^A_{\lambda - i\epsilon} \right] .
\end{align}
Denote by $ \delta^A_\lambda \equiv \delta(A - \lambda) \equiv P^A_\lambda$, $\lambda \in \sigma(A)$ the spectral projection defined by the strong limit
\begin{align}
	\delta^A_\lambda :=& \frac{1}{2\pi i} \slim_{\epsilon \searrow 0} \left[  R^A_{\lambda + i\epsilon} - R^A_{\lambda - i\epsilon} \right] .
\end{align}
If $\lambda$ is instead an element of the essential spectrum of $A$ one has that $\mathcal{PV}^A_\lambda, \delta^A_\lambda$ are defined by weak limits. If and only if the essential spectrum of $A$ is absolutely continuous then it is the case that $\mathrm{d}\mu^A_{\mathrm{e}}(\lambda) = \delta^A_\lambda\ \mathrm{d}\lambda$, where $\mathrm{d}\mu^A_{\mathrm{e}}(\lambda)$ is the essential spectral measure of $A$ and $\mathrm{d}\lambda$ is the Lebesgue measure on $\sigma_{\mathrm{e}}(A)$.
\end{defn}

The above definition permits the useful representation $\delta^{A}_{\lambda} = w^{A}_{\lambda}\phi^{A}_{\lambda} \otimes \phi^{A,*}_{\lambda}$ for $A$ of generalized multiplicity 1. One may observe through the spectral representation of $R^{L_{0}}_{z}$ that $\mathcal{PV}\psi^{L_{0}}_\lambda = \mathcal{PV}^{L_{0}}_\lambda \chi_{0}$ and that $\mathcal{PV}\xi^{L_{0}}_\lambda = \mathcal{PV}\psi^{L_{0}}_\lambda - \mathcal{PV}\psi^{L_{0}}_\lambda(0)\phi^{L_{0}}_\lambda = \xi^{L_{0}}_\lambda$, $\forall \lambda \in \sigma(L_{0})$, and analogously so for other operators.

We recall the method of spectral shifts as applied to rank-1 perturbations, see e.g. \cite{rank one}, for operators of the form specified by the $A_0, P, A = A_{0} - qP$ considered above. Through the resolvent formula it follows that
\begin{align}
	R^{A}_z &= R^{A_0}_z + R^{A_0}_zqP R^{A}_z \quad \Rightarrow \quad P R^{A}_z = P R^{A_0}_z + f^{A_0}_zqP R^{A}_z \\
		\Rightarrow \quad P R^{A}_z &=  (1 - qf^{A_0}_z)^{-1}P R^{A_0}_z \quad \Rightarrow \quad R^{A}_z = R^{A_0}_z + (1-qf^{A_{0}}_{z})^{-1}R^{A_0}_zqP R^{A_0}_z .
\end{align}
For $A$ essentially self-adjoint one may apply the definitions of $\mathcal{PV}^A_\lambda$ and $\delta^A_\lambda$ and find the corresponding shifts to $\mathcal{PV}^{A_{0}}_\lambda$ and $\delta^{A_{0}}_\lambda$. For $\lambda \in \sigma(A)$ it follows that
\begin{align}
	\mathcal{PV}^A_\lambda &= \mathcal{PV}^{A_0}_\lambda + g^{A_{0}}_{\lambda}[ (1-q\mathcal{PV}f^{A_{0}}_{\lambda})(\mathcal{PV}^{A_0}_\lambda qP \mathcal{PV}^{A_0}_\lambda - \pi^2 \delta^{A_0}_\lambda qP \delta^{A_0}_\lambda) \\
		&\quad\quad - \pi^2 q\delta f^{A_{0}}_{\lambda} (\mathcal{PV}^{A_0}_\lambda qP \delta^{A_0}_\lambda + \delta^{A_0}_\lambda qP \mathcal{PV}^{A_0}_\lambda) ]
\end{align}
\begin{align}
	\delta^A_\lambda &= \delta^{A_0}_\lambda + g^{A_{0}}_{\lambda} [ (1-q\mathcal{PV}f^{A_{0}}_{\lambda})(\mathcal{PV}^{A_0}_\lambda qP \delta^{A_0}_\lambda + \delta^{A_0}_\lambda qP \mathcal{PV}^{A_0}_\lambda) \\
		&\quad\quad + q\delta f^{A_{0}}_{\lambda} (\mathcal{PV}^{A_0}_\lambda qP \mathcal{PV}^{A_0}_\lambda - \pi^2 \delta^{A_0}_\lambda qP \delta^{A_0}_\lambda) ] ,
\end{align}
where
\begin{align}
	g^{A_{0}}_{\lambda} := [ (1-q\mathcal{PV}f^{A_{0}}_{\lambda})^2 + (q \pi \delta f^{A_{0}}_{\lambda})^2]^{-1} .
\end{align}

\begin{proof}[Proof of Theorem \ref{snthm01} Part (3)]
Here $\chi_{0}$ is both involved in the definition of important components of the normalization of of the $\phi_{\lambda}$ as well as the perturbation of $L_{0}$ to $L$. This will greatly simplify the expressions produced by the perturbation. By the definition of the resolvent function it is the case that $\mathcal{PV}f^{L_{0}}_{\lambda} = \mathcal{PV}\psi^{L_{0}}_{\lambda}(0)$ and $\delta f^{L_{0}}_{\lambda} = w^{L_{0}}_{\lambda}$. One may find that
\begin{align}
	R^{L}_z\chi_{0} &= \psi^{L}_{z} = (1-qf_{z})^{-1}\psi_{z} \quad \Rightarrow \quad f^{L}_{z} = (1-qf_{z})^{-1}f_{z}, \\
	\mathcal{PV}f^{L}_{z} &= g_{\lambda}[\mathcal{PV}f_{\lambda} - q(\mathcal{PV}f_{\lambda})^{2}-q(\pi w_{\lambda})^{2}],\quad \delta f^{L}_{\lambda} = w^{L}_{\lambda} = g_{\lambda}w_{\lambda}, \\
	\mathcal{PV}\psi^{L}_{\lambda} &= g_{\lambda}[\mathcal{PV} f_{\lambda}\phi_{\lambda}-q\mathcal{PV} f_{\lambda}\xi_{\lambda}+\xi_{\lambda}-q(\mathcal{PV} f_{\lambda})^{2}\phi_{\lambda} - q(\pi w_{\lambda})^{2}\phi_{\lambda}], \\
	\delta \psi^{L}_{\lambda} &= w^{L}_{\lambda}\phi^{L}_{\lambda} = g_{\lambda}w_{\lambda}(\phi_{\lambda}+q\xi_{\lambda}) \quad \Rightarrow \quad \phi^{L}_{\lambda} = \phi_{\lambda}+q\xi_{\lambda} \\
	\delta^{L}_{\lambda} &= w^{L}_{\lambda}\phi^{L}_{\lambda} \otimes \phi^{L,*}_{\lambda} = g_{\lambda}w_{\lambda}(\phi_{\lambda}+q\xi_{\lambda}) \otimes (\phi^{*}_{\lambda}+q\xi^{*}_{\lambda}) \\
	g_{\lambda} :&= [ (1-q\mathcal{PV}f_{\lambda})^2 + (q \pi w_{\lambda})^2]^{-1}
\end{align}
\end{proof}

Since $\delta^{L_{0}}_{\lambda}$ is regular at the threshold of $\sigma(L_{0})$ it is the case that the analysis of the threshold behavior of $\delta^{L}_{\lambda}$ is strongly controlled by the threshold behavior of
\begin{align}
	g_{\lambda} = \{ [1- q e^{-\lambda} \mathcal{PV}E_1(-\lambda)]^2 + [\pi q e^{-\lambda}]^2 \}^{-1},
\end{align}
which satisfies $w^{L}_{\lambda} = g_{\lambda} w_{\lambda}$. In particular $g_{\lambda}$ exhibits dominating behavior near the threshold due to the logarithmic divergence of $\mathcal{PV}E_1(-\lambda)$ near $\lambda = 0$.

\section{Decay Estimates for $L_{0}$ and $L$}

The Mourre estimate, see e.g. \cite{Mourre}, has been proven for $L_{0}$ by Chen, Fr\"ohlich, and Walcher \cite{CFW}, in order to prove its the spectrum is absolutely continuous and equal to $[0,\infty)$. We want to study pointwise decay estimates in time, which requires knowledge of the asymptotic properties of the resolvent at thresholds. The Mourre estimates do not apply at thresholds so we will need to use alternative methods.

Local decay estimates for $L_{0}$ have been found by Durhuus and Gayral \cite{noncommutative scattering} in the context of more general noncommutative solitons (where $L_{0}$ corresponds their ``diagonal case with 2 noncommuting spatial coordinates''). They found an unweighted estimate of the form $||e^{-itL_{0}}v||_{\infty} \le c |t|^{-1}(1+\log|t|)||v||_{1}$ for $|t|\ge1$. We present an alternative approach, in the context of Jacobi operators, which enhances the local decay estimate for the free Schr\"odinger operator and provides integrable decay for rank one boundary perturbations thereof for the restricted class of radial systems with two noncommutative spatial coordinates. We find weighted estimates
$$|| W_{\kappa,\tau}e^{-itL_{0}}W_{\kappa,\tau}v ||_{\infty} < ct^{-1}|| v ||_{1}, \quad || W_{\kappa,\tau}e^{-itL}P^{L}_{\mathrm{e}}W_{\kappa,\tau}v ||_{\infty} = \mathcal{O}(t^{-1}\log^{-2}t)$$
for $t\nearrow\infty$.

\subsection{Weighted estimates for spectral vectors}

\begin{defn}
Let $\mathbb{S}_r, \mathbb{D}_r \subset \mathbb{C}$ be respectively the circle and the disc of radius $r>0$ centered at the origin and $u \in \mathscr{T}$ a formal sequence for which there exist constants $r,c>0$ for which $|u(x)|r^{-x} < c$ for all $x \in \mathbb{Z}_+$. The \emph{generating function} of $u$ is the function $\zeta(u,\cdot) : \mathbb{S}_{r'} \to \mathbb{C}$ defined by $\zeta(u,s) := \sum_{x=0}^\infty u(x)s^x$, where $r' < r$. This permits the presentation of $u$ via
\begin{align}
	u(x) = \oint_\gamma \!\mathrm{d}s\ (2\pi i s)^{-1}s^{-x} \zeta(u,s),
\end{align}
where $\gamma$ is any positively oriented simple closed curve in $\mathbb{D}_r$ which encloses and does not pass through the origin.
\end{defn}

The Laguerre polynomials have the well-known generating function \cite{generating function}
\begin{align}
	\zeta(\phi_{\lambda},s):= \sum_{x=0}^\infty\phi_\lambda(x)s^x = (1-s)^{-1}\exp[-(1-s)^{-1}s \lambda],\quad |s|<1 .
\end{align}
We will also employ the notion of a reduced generating function.
\begin{defn}
For a given generating function $\zeta(v,s)$ of a vector $v$ let the \emph{reduced generating function} be the function $\widehat{\zeta}(v,s) := (1-s)\zeta(v,s)$.
\end{defn}
\noindent For example we have that $\widehat{\zeta}(\phi_{\lambda},s) = \exp[-(1-s)^{-1}s \lambda]$.

\begin{defn}
For $s \in \mathbb{C}$ we let
\begin{align}
s = re^{i\theta},\quad \widehat{s} := (1-s)^{-1}s,\quad \epsilon := (x + \kappa)^{-1}, \quad \widehat{\epsilon} := - 2^{-1} + 4^{-1}\epsilon , \quad \kappa > 0.
\end{align}
Should many variables $s_{j}$ be present we will use $r_j, \theta_j, \epsilon_j$ correspondingly.
\end{defn}

We are primarily concerned with estimates of operators in generating function presentation. In such forms one finds line integrations over dummy complex variables, $s_{j}$, with \emph{a priori} separate sums for each $x_{j}\in \mathbb{Z}_{+}$. One is therefore permitted to make the associated contours dependent on $x_{j}$. We then hereafter take
\begin{align}
r \set 1 - \epsilon .
\end{align}
 One may observe that
\begin{align}
|(1-s)^{-1}| \le x + \kappa,\quad |\widehat{s}| \le x + \kappa - 1 < x + \kappa
\end{align}
as well as the crucial estimate
\begin{align}
|s^{-x}| < e.
\end{align}
One is then permitted to work with polynomially weighted spaces instead of exponentially weighted ones.

\begin{lem}
One has that $|\exp(-\widehat{s}\lambda)| \le \exp(-\widehat{\epsilon}\lambda)$ for sufficiently large $\kappa > 0$.
\end{lem}

\begin{proof}
Let $r = 1 - \varepsilon$. For $s \in \mathbb{S}_r$ it must be the case that $\Re\widehat{s} = |(1-s)^{-1}|^2(|s|\cos\theta - |s|^2)$ attains its maximum value for $\Re s \le 0$.

\begin{align}
|1 - s|^2 = 1 - 2(1 - \epsilon)\cos\theta + (1 + \epsilon)^2 = m - m\varepsilon + \epsilon^2,
\end{align}
where $m := 2(1 - \cos\theta)$. For $\Re s \le 0$ one has that $2 \le m \le 4$ so $m$ is $\mathcal{O}(1)$. Then for $\Re s \le 0$ one has
\begin{align}
\Re s &= (m - m\epsilon + \epsilon^2)^{-1}[-2^{-1}m + (1 + 2^{-1}m)\epsilon + \epsilon^2] \\
&= [1 + \epsilon + (1-m^{-1})\epsilon^2 + \mathcal{O}(\epsilon^3)] [-2^{-1} + (m^{-1} + 2^{-1})\epsilon + \epsilon^2] \\
&= -2^{-1} + m^{-1}\epsilon + (3\cdot 2^{-1} + m^{-1})\epsilon^2 + \mathcal{O}(\epsilon)
\end{align}
and thereby
\begin{align}
|\exp(-\widehat{s}\lambda)| = \exp(-\Re\widehat{s}\lambda) \le \exp(2^{-1}\lambda - m^{-1}\epsilon \lambda) \le \exp(-\widehat{\epsilon}\lambda) .
\end{align}
\end{proof}

\begin{lem}
One has that $(\epsilon_{1} + \epsilon_{2})^{-1} < 4^{-1}(x_{1} + \kappa)(x_{2} + \kappa)$ for $\kappa > 1$.
\end{lem}

\begin{proof}
\begin{align}
(\epsilon_1 + \epsilon_2)^{-1} &= [(x_1 + \kappa)^{-1} + (x_2 + \kappa)^{-1}]^{-1}\\
&= [(x_1 + \kappa) + (x_2 + \kappa)]^{-1}(x_1 + \kappa)(x_2 + \kappa) \\
&< 4^{-1}(x_1 + \kappa)(x_2 + \kappa),
\end{align}
for sufficiently large $\kappa > 1$.
\end{proof}

\begin{lem}
One has the representation
\begin{align}
	\xi_{z}(x) = \int_0^\infty \!\mathrm{d}\eta\ e^{-\eta} (\eta - z)^{-1} [\phi_{\eta}(x) - \phi_z(x)].
\end{align}
\end{lem}

\begin{proof}
By the spectral representation of $R^{L_0}_z$ it is the case that \\ $\psi_{z}(x) = \int_0^\infty \! \mathrm{d}\lambda\ e^{-\lambda} (\lambda - z)^{-1} \phi_{\lambda}(x)$, where we have used the normalization condition $\phi_{\lambda}(0) = 1$, $\forall \lambda \in \sigma(L_0)$.
\end{proof}

\begin{lem}
One has the generating function representation
\begin{align}
	\xi_{z}(x) = \oint_{\mathbb{S}_{r}} \!\mathrm{d}s\ (2\pi i s)^{-1}s^{-x} \zeta(\xi_{z},s),\quad \forall z \in \mathbb{C}
\end{align}
where
\begin{align}
\zeta(\xi_{z},s) = (1-s)^{-1} \int_0^\infty \!\mathrm{d}\eta\ e^{-\eta} K(\eta, z, s)
\end{align}
and
\begin{align}
K(\eta, z, s) := (\eta - z)^{-1}\left[ \exp(-\widehat{s}\eta) - \exp(-\widehat{s}z) \right]
\end{align}
for $\eta \in \mathbb{R}_{+}$, $z \in \mathbb{C}$, $s \in \mathbb{S}_{r}$.

\end{lem}

\begin{proof}
First, consider $z \in \mathbb{C} \setminus \mathbb{R}_{+} =: \Sigma$. Since $|s| < 1$ there exists a $c > 0$ such that $|\widehat{s}| < c$. It follows that
\begin{align}
\left| K(\eta, z, s) \right| &\le \left| (\eta - z)^{-1} \right| \left[ \left| \exp(-\widehat{s}\eta) \right| + \left| \exp(-\widehat{s}z) \right| \right] \\
&\le \left[ \mathrm{dist}(\Sigma, z) \right]^{-1} \left[ \exp(-\Re\widehat{s}\eta) + \exp(-c|z|) \right] < \infty .
\end{align}
Second, let $z \equiv \lambda \in \mathbb{R}_{+}$. By mean value theorem one has
\begin{align}
	K(\eta, \lambda, s) &= (\eta - \lambda)^{-1}\left[ \Re\exp(-\widehat{s}\eta) - \Re\exp(-\widehat{s}z) \right] \\
		&\quad\quad+ i(\eta - \lambda)^{-1}\left[ \Im\exp(-\widehat{s}\eta) - \Im\exp(-\widehat{s}z) \right] \\
	&= \mathrm{d}_{\eta}\Re\exp(-\widehat{s}\eta)\lfloor_{\eta = \mu_{1}} + i\mathrm{d}_{\eta}\Im\exp(-\widehat{s}\eta)\lfloor_{\eta = \mu_{2}} \\
	&= {1 \over 2}\left[ (-\widehat{s})\exp(-\widehat{s}\mu_{1}) + (-\overline{\widehat{s}})\exp(-\overline{\widehat{s}}\mu_{1}) \right. \\
		&\quad\quad \left.+ (-\widehat{s})\exp(-\widehat{s}\mu_{2}) - (-\overline{\widehat{s}})\exp(-\overline{\widehat{s}}\mu_{2}) \right] ,
\end{align}
where $\mu_{j} \equiv \mu_{j}(r,\theta, \eta, \lambda) \in [\min(\eta, \lambda), \max(\eta, \lambda)]$. Then
\begin{align}
	|K(\eta, \lambda, s)| &\le {1 \over 2}\left[ |\widehat{s}| |\exp(-\widehat{s}\mu_{1})| + |\widehat{s}| |\exp(-\widehat{s}\mu_{1})| \right. \\
		&\quad\quad \left. + |\widehat{s}| |\exp(-\widehat{s}\mu_{2})| + |\widehat{s}| |\exp(-\widehat{s}\mu_{2})|  \right] \\
	&= |\widehat{s}| \left[ \exp(-\Re\widehat{s}\mu_{1}) + \exp(-\Re\widehat{s}\mu_{2}) \right] \\
		&\le 2|\widehat{s}| \exp[-\widehat{\epsilon}(\eta + \lambda)] < \infty.
\end{align}
One may observe that
\begin{align}
\int_0^\infty \!\mathrm{d}\eta\ e^{-\eta} |K(\eta, z, s)| &\le \int_0^\infty \!\mathrm{d}\eta\ e^{-\eta} 2|\widehat{s}| \exp[-\widehat{\epsilon}(\eta + \lambda)] \\
&= 2|\widehat{s}| \exp(-\widehat{\epsilon}\lambda) \int_0^\infty \!\mathrm{d}\eta\ \exp[-(1+\widehat{\epsilon})\eta] \\
&= 2|\widehat{s}| \exp(-\widehat{\epsilon}\lambda) (1+\widehat{\epsilon})^{-1} < \infty
\end{align}
The multi-integral of the generating function representation of $\xi_{z}(x)$ converges absolutely and thereby Fubini's theorem permits
\begin{align}
\xi_{z}(x) &= \int_0^\infty \!\mathrm{d}\eta\ e^{-\eta} (\eta - z)^{-1} [\phi_\eta(x) - \phi_z(x)] \\
&=  \int_0^\infty \!\mathrm{d}\eta\ e^{-\eta}  \oint_{\mathbb{S}_{r}} \!\mathrm{d}s\ (2\pi i s)^{-1}s^{-x} (1-s)^{-1} K(\eta, z, s) \\
&= \oint_{\mathbb{S}_{r}} \!\mathrm{d}s\ (2\pi i s)^{-1}s^{-x} (1-s)^{-1} \int_0^\infty \!\mathrm{d}\eta\ e^{-\eta} K(\eta, z, s)
\end{align}
for all $z \in \mathbb{C}$.
\end{proof}

\begin{lem}
For sufficiently large $\kappa > 0$, it is the case that
\begin{align}
	|\mathrm{d}_{\lambda}^{n}\widehat{\zeta}(\phi_{\lambda}, s)| < \widehat{c}(\phi_{\lambda},n) \exp(-\widehat{\epsilon}\lambda),\quad n \in \mathbb{Z}_{+},
\end{align}
where $\widehat{c}(\phi_{\lambda},n) := (x+\kappa)^{n}$ and
\begin{align}
	|\mathrm{d}_{\lambda}^{n}\widehat{\zeta}(\xi_{\lambda}, s)| < \widehat{c}(\xi_{\lambda},n) \exp(-\widehat{\epsilon}\lambda),\quad n = 0, 1, 2,
\end{align}
where
\begin{align}
	\widehat{c}(\xi_{\lambda},0) := 2(x + \kappa),\quad \widehat{c}(\xi_{\lambda},1) := 4(x + \kappa),\quad \widehat{c}(\xi_{\lambda},2) := 4(x + \kappa)^{2} .
\end{align}
\end{lem}

\begin{proof}
\underline{For $\phi_{\lambda}$:}
\begin{align}
	|\mathrm{d}_{\lambda}^{n}\widehat{\zeta}(\phi_{\lambda}, s)| &= |\mathrm{d}_{\lambda}^{n}\exp(-\widehat{s}\lambda)| = |\widehat{s}^{n}\exp(-\widehat{s}\lambda)| = |\widehat{s}|^{n} |\exp(-\widehat{s}\lambda)| \\
	&\le |\widehat{s}|^{n} \exp(-\widehat{\epsilon}\lambda) < (x+\kappa)^{n} \exp(-\widehat{\epsilon}\lambda)
\end{align}
\underline{For $\xi_{\lambda}$:}
One may observe that $K(\eta, \lambda, s) = K(\lambda, \eta, s)$. Then, by integration by parts, one has
\begin{align}
	\mathrm{d}_{\lambda}\widehat{\zeta}(\xi_{\lambda},s) &= \int_{0}^{\infty}\!\mathrm{d}\eta\ e^{-\eta} \mathrm{d}_{\lambda}K(\eta, \lambda, s) = \int_{0}^{\infty}\!\mathrm{d}\eta\ e^{-\eta} \mathrm{d}_{\eta}K(\eta, \lambda, s) \\
		&= - K(0, \lambda, s) + \widehat{\zeta}(\xi_{\lambda},s)
\end{align}
and thereby
\begin{align}
	\mathrm{d}_{\lambda}^{n}\widehat{\zeta}(\xi_{\lambda},s) = -\sum_{k = 0}^{n-1}\mathrm{d}_{\lambda}K(0, \lambda, s) + \widehat{\zeta}(\xi_{\lambda},s),
\end{align}
where the sum is defined to vanish when the upper bound is negative. It follows that
\begin{align}
	|K(\eta, \lambda, s)| \le 2|\widehat{s}| \exp[-\widehat{\epsilon}(\eta + \lambda)],\quad |\widehat{\zeta}(\xi_{\lambda},s)| \le 2|\widehat{s}|(1+\widehat{\epsilon})^{-1} \exp( -\widehat{\epsilon} \lambda) .
\end{align}
Consider an arbitrary $f \in C^{2}(\mathbb{R},\mathbb{R})$ and let $f_{*}$ be its Newton quotient so that
\begin{align}
	f_{*}(a_{0}, a) := (a_{0} - a)^{-1}[f(a_{0}) - f(a)] .
\end{align}
One has by mean value theorem
\begin{align}
	\mathrm{d}_{a}f_{*}(a_{0}, a) &= (a_{0} - a)^{-2}[f(a_{0}) - f(a) - (a_{0} - a)\mathrm{d}_{a} f(a)] \\
		&= (a_{0} - a)^{-1}[\mathrm{d}_{a_{1}} f(a_{1}) - \mathrm{d}_{a} f(a)] ,\quad a_{1} \in [\min(a_{0}, a), \max(a_{0}, a)] \\
		&= (a_{0} - a)^{-1}(a_{1} - a)\mathrm{d}^{2}_{a_{2}} f(a_{2}), \quad a_{2} \in [\min(a_{1}, a), \max(a_{1}, a)]
\end{align}
\begin{align}
	\Rightarrow \quad |\mathrm{d}_{a}f_{*}(a_{0}, a)| &\le |(a_{0} - a)^{-1}| |a_{1} - a| |\mathrm{d}^{2}_{a_{2}} f(a_{2}) | \le |\mathrm{d}^{2}_{a_{2}} f(a_{2}) | .
\end{align}

Let $(\Re, \Im)z$ be a presentation for the real and imaginary parts of $z \in \mathbb{C}$ whose ordering in compatible with the respective ordering of $\pm$. Let $i_{+} := 1$, $i_{-} := i$ and $\mu_{\pm} \in [\min(\eta, \lambda), \max(\eta, \lambda)] $. It follows that
\begin{align}
	|\mathrm{d}_{\lambda}(\Re, \Im)K(\eta, \lambda, s)| &\le |\mathrm{d}_{\lambda}^{2}(\Re, \Im)\exp(-\widehat{s}\lambda)| \lfloor_{\lambda = \mu_{\pm}} \\
		&= |\mathrm{d}_{\lambda}^{2}(2i_{\pm})^{-1}[\exp(-\widehat{s}\lambda) \pm \exp(-\overline{\widehat{s}}\lambda)]| \lfloor_{\lambda = \mu_{\pm}} \\
		&\le |\widehat{s}|^{2}\exp[-\widehat{\epsilon}(\eta + \lambda)] \\
\Rightarrow \quad |\mathrm{d}_{\lambda}K(\eta, \lambda, s)| &\le 2|\widehat{s}|^{2}\exp[-\widehat{\epsilon}(\eta + \lambda)] .
\end{align}
Then
\begin{align}
	| \widehat{\zeta}(\xi_{\lambda},s)| &\le 2|\widehat{s}|(1+\widehat{\epsilon})^{-1} \exp( -\widehat{\epsilon} \lambda) < 2(x + \kappa)\exp( -\widehat{\epsilon} \lambda) \\
	| \mathrm{d}_{\lambda}\widehat{\zeta}(\xi_{\lambda},s)| &\le 2|\widehat{s}|[(1+\widehat{\epsilon})^{-1} + 1] \exp( -\widehat{\epsilon} \lambda) < 4(x + \kappa)\exp( -\widehat{\epsilon} \lambda) \\
	| \mathrm{d}_{\lambda}^{2}\widehat{\zeta}(\xi_{\lambda},s)| &\le 2|\widehat{s}|[(1+\widehat{\epsilon})^{-1} + 1 + |\widehat{s}|] \exp( -\widehat{\epsilon} \lambda) < 4(x + \kappa)^{2} \exp( -\widehat{\epsilon} \lambda) .
\end{align}
\end{proof}

\begin{rem}
If estimates of $\mathrm{d}_{\lambda}^{n}\widehat{\zeta}(\xi_{\lambda},s)$ for $2 < n \in \mathbb{Z}$ were required the above method would not follow so straightforwardly due to the inapplicability of the mean value theorem for yet higher derivatives.
\end{rem}

\begin{cor}\label{cor01}
For sufficiently large $\kappa > 0$, one has that
\begin{align}
	\left|\mathrm{d}^{n}_{\lambda}\left[w^{1/2}_{\lambda}\phi_{\lambda}(x)\right]\right| < c(\phi_{\lambda},n)\exp(-4^{-1}\epsilon\lambda),\quad n = 0, 1, 2,
\end{align}
where $c(\phi_{\lambda},n) := 3^{n+1}(x+\kappa)^{n+1}$ and
\begin{align}
	\left|\mathrm{d}^{n}_{\lambda}\left[w^{1/2}_{\lambda}\xi_{\lambda}(x)\right]\right| < c(\xi_{\lambda},n)\exp(-4^{-1}\epsilon\lambda),\quad n = 0, 1, 2,
\end{align}
where
\begin{align}
	c(\xi_{\lambda},0) := 6(x+\kappa)^{2},\quad c(\xi_{\lambda},1) := 15(x+\kappa)^{2},\quad c(\xi_{\lambda},2) := 21(x+\kappa)^{3}
\end{align}
\end{cor}

\begin{proof}
\underline{For $\phi_{\lambda}$:}
\begin{align}
	\left|\mathrm{d}^{n}_{\lambda}\left[w^{1/2}_{\lambda}\phi_{\lambda}(x)\right]\right| &= \left| \oint_{\mathbb{S}_{r}} \mathrm{d}s\ (2\pi i s)^{-1}s^{-x} (1-s)^{-1} \mathrm{d}^{n}_{\lambda}\left[w^{1/2}_{\lambda}\widehat{\zeta}(\phi_{\lambda},s)\right] \right| , \\
		&= \left| \oint_{\mathbb{S}_{r}} \mathrm{d}s\ (2\pi i s)^{-1}s^{-x} (1-s)^{-1} \right. \\
			&\quad\quad \times \left. \sum_{k=0}^{n}\binom{n}{k}\mathrm{d}^{n-k}_{\lambda}w^{1/2}_{\lambda}\mathrm{d}^{k}_{\lambda}\widehat{\zeta}(\phi_{\lambda},s) \right| , \\
		&\le \oint_{\mathbb{S}_{r}} \left| \mathrm{d}s\ (2\pi i s)^{-1} \right| \left| s^{-x} \right| \left| (1-s)^{-1} \right| \\
			&\quad\quad \times \sum_{k=0}^{n}\binom{n}{k} \left| \mathrm{d}^{n-k}_{\lambda}w^{1/2}_{\lambda} \right| \left| \mathrm{d}^{k}_{\lambda}\widehat{\zeta}(\phi_{\lambda},s) \right| ,
\end{align}
\begin{align}
		&< (1)(3)(x+\kappa) \sum_{k=0}^{n}\binom{n}{k} 2^{-(n-k)}w^{1/2}_{\lambda}\widehat{c}(\phi_{\lambda},k)\exp(-\widehat{\epsilon}\lambda) \\
		&= 3(x+\kappa) \sum_{k=0}^{n}\binom{n}{k} 2^{-(n-k)}\widehat{c}(\phi_{\lambda},k)\exp(-4^{-1}\epsilon\lambda) .
\end{align}
For $n=0$:
\begin{align}
	\left|w^{1/2}_{\lambda}\phi_{\lambda}(x)\right| < 3(x+\kappa) \widehat{c}(\phi_{\lambda},0)\exp(-4^{-1}\epsilon\lambda) = 3(x+\kappa) \exp(-4^{-1}\epsilon\lambda) .
\end{align}
For $n = 1$:
\begin{align}
	\left|\mathrm{d}_{\lambda}\left[w^{1/2}_{\lambda}\phi_{\lambda}(x)\right]\right| &< 3(x+\kappa) \left[ 2^{-1}\widehat{c}(\phi_{\lambda},0) + \widehat{c}(\phi_{\lambda},1) \right] \exp(-4^{-1}\epsilon\lambda) \\
		&= 3(x+\kappa) \left[ 2^{-1} + (x+\kappa) \right] \exp(-4^{-1}\epsilon\lambda) \\
		&< 6(x+\kappa)^{2} \exp(-4^{-1}\epsilon\lambda) .
\end{align}
For $n = 2$:
\begin{align}
	\left|\mathrm{d}^{2}_{\lambda}\left[w^{1/2}_{\lambda}\phi_{\lambda}(x)\right]\right| &< 3(x+\kappa) \left[ 2^{-2}\widehat{c}(\phi_{\lambda},0) + 2^{-1}\widehat{c}(\phi_{\lambda},1) + \widehat{c}(\phi_{\lambda},2) \right] \\
			&\quad\quad \times \exp(-4^{-1}\epsilon\lambda) \\
		&= 3(x+\kappa) \left[ 2^{-2} + 2^{-1}(x+\kappa) + (x+\kappa)^{2} \right] \\
			&\quad\quad \times \exp(-4^{-1}\epsilon\lambda) \\
		&< 9(x+\kappa)^{3} \exp(-4^{-1}\epsilon\lambda) .
\end{align}

\underline{For $\xi_{\lambda}$:}
\begin{align}
	\left|w^{1/2}_{\lambda}\xi_{\lambda}(x)\right| < 3(x+\kappa) \sum_{k=0}^{n}\binom{n}{k} 2^{-(n-k)}\widehat{c}(\xi_{\lambda},k)\exp(-4^{-1}\epsilon\lambda)
\end{align}
For $n = 0$:
\begin{align}
	\left|w^{1/2}_{\lambda}\xi_{\lambda}(x)\right| < 3(x+\kappa) \widehat{c}(\xi_{\lambda},0)\exp(-4^{-1}\epsilon\lambda) = 6(x+\kappa)^{2} \exp(-4^{-1}\epsilon\lambda) .
\end{align}
For $n = 1$:
\begin{align}
	\left|\mathrm{d}_{\lambda}\left[w^{1/2}_{\lambda}\xi_{\lambda}(x)\right]\right| &< 3(x+\kappa) \left[ 2^{-1}\widehat{c}(\xi_{\lambda},0) + \widehat{c}(\xi_{\lambda},1) \right] \exp(-4^{-1}\epsilon\lambda) \\
		&= 3(x+\kappa) \left[ (x+\kappa) + 4(x+\kappa) \right] \exp(-4^{-1}\epsilon\lambda) \\
		&= 15(x+\kappa)^{2} \exp(-4^{-1}\epsilon\lambda) .
\end{align}
For $n = 2$:
\begin{align}
	\left|\mathrm{d}^{2}_{\lambda}\left[w^{1/2}_{\lambda}\xi_{\lambda}(x)\right]\right| &< 3(x+\kappa) \left[ 2^{-2}\widehat{c}(\xi_{\lambda},0) + 2^{-1}\widehat{c}(\xi_{\lambda},1) + \widehat{c}(\xi_{\lambda},2) \right] \\
			&\quad\quad \times \exp(-4^{-1}\epsilon\lambda) \\
		&= 3(x+\kappa) \left[ 2^{-1}(x+\kappa) + 2(x+\kappa) + 4(x+\kappa)^{2} \right] \\
			&\quad\quad \times \exp(-4^{-1}\epsilon\lambda) \\
		&< 21(x+\kappa)^{3} \exp(-4^{-1}\epsilon\lambda) .
\end{align}
\end{proof}

\subsection{Local time decay for $L_{0}$}

\begin{proof}[Proof of Theorem \ref{snthm02}]
Let $t > 0$, $|s_{j}| = r_{j} = 1 - (x_{j}+\kappa)^{-1}$, $j=1,2$, and $1 < \kappa \in \mathbb{R}$ be a sufficiently large constant. It is the case that
\begin{align}
	\left| e^{-itL_{0}}(x_{1},x_{2}) \right| &= \left| \int_{0}^{\infty}\mathrm{d}\lambda\ e^{- it\lambda}w_{\lambda}\phi_{\lambda}(x_{1})\phi_{\lambda}(x_{2}) \right| \\
		&= \left| \int_{0}^{\infty}\mathrm{d}\lambda\ (-it)^{-1}\mathrm{d}_{\lambda}e^{- it\lambda} \left[ w^{1/2}_{\lambda}\phi_{\lambda}(x_{1}) \right] \left[ w^{1/2}_{\lambda} \phi_{\lambda}(x_{2}) \right] \right| \\
		&\le \left| -(-it)^{-1} - \int_{0}^{\infty}\mathrm{d}\lambda\ (-it)^{-1}e^{- it\lambda} \left\{ \mathrm{d}_{\lambda} \left[ w^{1/2}_{\lambda}\phi_{\lambda}(x_{1}) \right] \right.\right. \\
			&\quad\quad \left.\left. \times \left[ w^{1/2}_{\lambda} \phi_{\lambda}(x_{2}) \right] + \left[ w^{1/2}_{\lambda}\phi_{\lambda}(x_{1}) \right] \mathrm{d}_{\lambda} \left[ w^{1/2}_{\lambda} \phi_{\lambda}(x_{2}) \right] \right\} \right| \\
		&\le t^{-1}\left( 1 + \int_{0}^{\infty}\mathrm{d}\lambda\ \left\{ \left| \mathrm{d}_{\lambda} \left[ w^{1/2}_{\lambda}\phi_{\lambda}(x_{1}) \right] \right| \ \left| w^{1/2}_{\lambda} \phi_{\lambda}(x_{2}) \right| \right.\right. \\
			&\quad\quad\left.\left. + \left| w^{1/2}_{\lambda}\phi_{\lambda}(x_{1}) \right| \ \left| \mathrm{d}_{\lambda} \left[ w^{1/2}_{\lambda} \phi_{\lambda}(x_{2}) \right] \right| \right\} \right)
\end{align}
\begin{align}
		&< t^{-1}\{ 1 + \int_{0}^{\infty}\mathrm{d}\lambda\ [ (6)(x_{1} + \kappa)^{2}\exp(-4^{-1}\epsilon_{1}\lambda)\\
		&\qquad\times (3)(x_{2} + \kappa)\exp(-4^{-1}\epsilon_{2}\lambda) + (3)(x_{1} + \kappa)\exp(-4^{-1}\epsilon_{1}\lambda)\\
		&\qquad\times (6)(x_{2} + \kappa)\exp(-4^{-1}\epsilon_{2}\lambda) ] \} \\
		&\le t^{-1}\left\{ 1 + 18(x_{1}+\kappa)^{2}(x_{2}+\kappa)^{2}\int_{0}^{\infty}\mathrm{d}\lambda\ \exp[-4^{-1}(\epsilon_{1} +\epsilon_{2})\lambda] \right\} \\
		&= t^{-1}\left[ 1 + 288(x_{1}+\kappa)^{2}(x_{2}+\kappa)^{2}(\epsilon_{1} +\epsilon_{2})^{-1} \right] \\
		&< t^{-1}\left[ 1 + 72(x_{1}+\kappa)^{3}(x_{2}+\kappa)^{3} \right] \\
		&\le 73(x_{1}+\kappa)^{3}(x_{2}+\kappa)^{3} t^{-1} .
\end{align}
\end{proof}

\subsection{Local time decay for $L$}

We recall without proof Lemma 3.12 from \cite{2D}:

\begin{unlem}
Let $\mathscr{B}$ be a Banach space and $\lambda_{+} > \lambda_{-}$ be real constants. If $F(\lambda)$ has the properties
\begin{enumerate}
\item $F \in C(\lambda_{-},\lambda_{+};\mathscr{B})$
\item $F(\lambda_{-}) = F(\lambda) = 0 , \quad \lambda > \lambda_{+}$
\item $\mathrm{d}_{\lambda} F \in L^1(\lambda_{-} + \delta, \lambda_{+};\mathscr{B}) , \quad \forall \delta>0$
\item $\mathrm{d}_{\lambda} F(\lambda) = \mathcal{O}( [\lambda - \lambda_{-}]^{-1}\log^{-3}[\lambda - \lambda_{-}] ) , \quad \lambda \searrow \lambda_{-}$
\item $\mathrm{d}_{\lambda}^{2} F(\lambda) = \mathcal{O}( [\lambda-\lambda_{-}]^{-2}\log^{-2}[\lambda - \lambda_{-}] ) , \quad \lambda \searrow \lambda_{-}$
\end{enumerate}
then
\begin{align}
\int_{\lambda_{-}}^\infty \!\mathrm{d}\lambda\ e^{-it\lambda} F(\lambda) = \mathcal{O}(t^{-1}\log^{-2}t),\quad t\nearrow\infty
\end{align}
in the norm of $\mathscr{B}$.
\end{unlem}

The proof of Theorem \ref{snthm03} requires the spectral representation $e^{-itL} P_\mathrm{e} L = \int_{\sigma_\mathrm{e}(L)} \mathrm{d}\lambda\ e^{-it\lambda} \lambda \delta^L_\lambda$ and in turn the weighted estimates of the essential spectral measure found previously. These methods follow from the principle of asymptotics extended from the scalar Laplace transform to the context of spectral calculus: the long time behavior of solutions is given by the threshold behavior of the resolvent of the Schr\"odinger operator which specifies the dynamics. The role of the Banach space defined above is to transfer the problem back to the more tractable realm of the scalar Laplace transform.

\begin{proof}[Proof of Theorem \ref{snthm03}]
Let $\mathscr{B} = \left\{ A \in \mathcal{L}(\mathscr{T}) : || A ||_{\mathscr{B}} < \infty \right\} $ be the Banach space complete in the norm
\begin{align}
|| A ||_{\mathscr{B}} := \sup_{v \in \ell^{1}} { || W_{\kappa,\tau}AW_{\kappa,\tau}v ||_{\infty} \over || v ||_{1} } .
\end{align}
Let $F(\lambda) = \delta^{L}_{\lambda}$. We will verify the appropriate properties of $F(\lambda)$ for $\lambda_{-} = 0$ and $\lambda _{+} = \infty$.

We recall that
\begin{align}
	F(\lambda,x_{1},x_{2}) &= w^{L}_{\lambda}\phi^{L}_{\lambda}(x_{1})\phi^{L}_{\lambda}(x_{2})\\
	&= g_{\lambda}w_{\lambda}[\phi_{\lambda}(x_{1}) + q\xi_{\lambda}(x_{1})][\phi_{\lambda}(x_{2}) + q\xi_{\lambda}(x_{2})] .
\end{align}
One may observe that
\begin{align}
	|\mathrm{d}^{n}_{\lambda}[w^{1/2}_{\lambda}\phi^{L}_{\lambda}(x)]| &\le |\mathrm{d}^{n}_{\lambda}[w^{1/2}_{\lambda}\phi_{\lambda}(x)]| +q |\mathrm{d}^{n}_{\lambda}[w^{1/2}_{\lambda}\xi_{\lambda}(x)]| \\
		&< [c(\phi_{\lambda},n) + qc(\xi_{\lambda},n)]\exp(-4^{-1}\epsilon\lambda), \\
		&< c(\phi^{L}_{\lambda},n)\exp(-4^{-1}\epsilon\lambda),\quad n = 0, 1, 2
\end{align}
where here we choose
\begin{align}
	c(\phi^{L}_{\lambda},0) &:= 3(1+3q)(x+\kappa)^{2}, \\
	c(\phi^{L}_{\lambda},1) &:= 6(1+3q)(x+\kappa)^{2}, \\
	c(\phi^{L}_{\lambda},2) &:= 9(1+3q)(x+\kappa)^{3}.
\end{align}
The logarithmic behavior of $\mathcal{PV}E_1(-\lambda)$ near $\lambda = 0$ is very important for many estimates. One may see by inspection that $g_{\lambda} := \{ [1- q e^{-\lambda} \mathcal{PV}E_1(-\lambda)]^2 + [\pi q e^{-\lambda}]^2 \}^{-1}$ has the properties:
\begin{align}
	g_{\lambda} &= | g_{\lambda} | \le \widehat{g}_{0}(q) < \infty, \quad \forall \lambda \in [0, \infty) \\
	| \mathrm{d}_{\lambda}g_{\lambda} | &\le \widehat{g}_{0}(q)\widehat{g}_{1}(q, \delta) < \infty,  \quad \forall \lambda \in [\delta, \infty) \\
	g_{0} &= g_{\infty} = 0 , \\
	g_\lambda &= \mathcal{O}(\log^{-2}\lambda), \quad \lambda \searrow 0 \\
	\mathrm{d}_{\lambda}g_{\lambda} &= \mathcal{O}(\lambda^{-1}\log^{-3}\lambda), \quad \lambda \searrow 0 \\
	\mathrm{d}_{\lambda}^{2}g_{\lambda} &= \mathcal{O}(\lambda^{-2}\log^{-3}\lambda), \quad \lambda \searrow 0 \\
		&= \mathcal{O}(\lambda^{-2}\log^{-2}\lambda)
\end{align}
where $0 < \widehat{g}_{0}(q), \widehat{g}_{1}(q, \delta) < \infty$ are constants whose other properties are not needed here. $g_{\lambda}$ is the only function of $\lambda$ involved in the definition of $F(\lambda)$ whose derivatives are unbounded in the neighborhood of the threshold $\lambda = 0$ and thereby the derivatives of $g_{\lambda}$ are dominant in determining the properties of the derivatives of $F(\lambda)$.

\noindent\underline{Properties (1), (2):}
One may observe that the properties follow by inspection.

\noindent\underline{Property (3):}
For $\lambda \in [\delta, \infty)$ one has
\begin{align}
	| \mathrm{d}_{\lambda}F(\lambda,x_{1},x_{2}) | &= | \mathrm{d}_{\lambda}\{ g_{\lambda}[w^{1/2}_{\lambda}\phi^{L}_{\lambda}(x_{1})][w^{1/2}_{\lambda}\phi^{L}_{\lambda}(x_{2})] \} | \\
		&\le |\mathrm{d}_{\lambda}g_{\lambda}| |[w^{1/2}_{\lambda}\phi^{L}_{\lambda}(x_{1})]| |[w^{1/2}_{\lambda}\phi^{L}_{\lambda}(x_{2})]| \\
			&\quad\quad + |g_{\lambda}||\mathrm{d}_{\lambda}[w^{1/2}_{\lambda}\phi^{L}_{\lambda}(x_{1})]| |[w^{1/2}_{\lambda}\phi^{L}_{\lambda}(x_{2})]| \\
			&\quad\quad + |g_{\lambda}| |[w^{1/2}_{\lambda}\phi^{L}_{\lambda}(x_{1})]| |\mathrm{d}_{\lambda}[w^{1/2}_{\lambda}\phi^{L}_{\lambda}(x_{2})]| \\
		&< \widehat{g}_{0}(q)\widehat{g}_{1}(q, \delta)(3)(1+3q)(x_{1}+\kappa)^{2}\exp(-4^{-1}\epsilon_{1}\lambda) \\
			&\quad\quad \times (3)(1+3q)(x_{2}+\kappa)^{2}\exp(-4^{-1}\epsilon_{1}\lambda) \\
			&\quad\quad + \widehat{g}_{0}(q)(6)(1+3q)(x_{1}+\kappa)^{2}\exp(-4^{-1}\epsilon_{1}\lambda) \\
			&\quad\quad \times (3)(1+3q)(x_{2}+\kappa)^{2}\exp(-4^{-1}\epsilon_{2}\lambda) \\
			&\quad\quad + \widehat{g}_{0}(q)(3)(1+3q)(x_{1}+\kappa)^{2}\exp(-4^{-1}\epsilon_{1}\lambda) \\
			&\quad\quad \times (6)(1+3q)(x_{2}+\kappa)^{2}\exp(-4^{-1}\epsilon_{2}\lambda) \\
		&= c_{0}(q,\delta)(x_{1}+\kappa)^{2}(x_{2}+\kappa)^{2}\exp[-4^{-1}(\epsilon_{1} +\epsilon_{2} )\lambda],
\end{align}
where $c_{0}(q,\delta)$ is a constant.

\noindent\underline{Property (4):}
For $\lambda \searrow 0$ one has
\begin{align}
	| \mathrm{d}_{\lambda}F(\lambda,x_{1},x_{2}) | &= | \mathrm{d}_{\lambda}\{ g_{\lambda}[w^{1/2}_{\lambda}\phi^{L}_{\lambda}(x_{1})][w^{1/2}_{\lambda}\phi^{L}_{\lambda}(x_{2})] \} | \\
		&< |\mathrm{d}_{\lambda}g_{\lambda}|(3)(1+3q)(x_{1}+\kappa)^{2}\exp(-4^{-1}\epsilon_{1}\lambda) \\
			&\quad\quad \times (3)(1+3q)(x_{2}+\kappa)^{2}\exp(-4^{-1}\epsilon_{1}\lambda) \\
			&\quad\quad +\widehat{g}_{0}(q)(6)(1+3q)(x_{1}+\kappa)^{2}\exp(-4^{-1}\epsilon_{1}\lambda) \\
			&\quad\quad \times (3)(1+3q)(x_{2}+\kappa)^{2}\exp(-4^{-1}\epsilon_{2}\lambda) \\
			&\quad\quad + \widehat{g}_{0}(q)(3)(1+3q)(x_{1}+\kappa)^{2}\exp(-4^{-1}\epsilon_{1}\lambda) \\
			&\quad\quad \times (6)(1+3q)(x_{2}+\kappa)^{2}\exp(-4^{-1}\epsilon_{2}\lambda) \\
		&\le c_{1}(q,\delta)(x_{1}+\kappa)^{2}(x_{2}+\kappa)^{2}\exp[-4^{-1}(\epsilon_{1} +\epsilon_{2} )\lambda] |\mathrm{d}_{\lambda}g_{\lambda}| \\
		 &= \mathcal{O}(\lambda^{-1}\log^{-3}\lambda)
\end{align}
in the norm of $\mathscr{B}$, where $c_{1}(q,\delta)$ is a constant.

\noindent\underline{Property (5):}
For $\lambda \searrow 0$ one has
\begin{align}
	| \mathrm{d}^{2}_{\lambda}F(\lambda,x_{1},x_{2}) | &= | \mathrm{d}^{2}_{\lambda}\{ g_{\lambda}[w^{1/2}_{\lambda}\phi^{L}_{\lambda}(x_{1})][w^{1/2}_{\lambda}\phi^{L}_{\lambda}(x_{2})] \} | \\
		&\le c_{2}(q,\delta)(x_{1}+\kappa)^{3}(x_{2}+\kappa)^{3}\exp[-4^{-1}(\epsilon_{1} +\epsilon_{2} )\lambda] |\mathrm{d}^{2}_{\lambda}g_{\lambda}| \\
		&= \mathcal{O}(\lambda^{-2}\log^{-2}\lambda)
\end{align}
in the norm of $\mathscr{B}$, where $c_{2}(q,\delta)$ is a constant.
\end{proof}

\appendix

\section{Semi-Analytic Vector Theorem}

We will review the statement and proof of the Semi-Analytic Vector Theorem as is presented in \cite{semi-analytic}.

\begin{defn}
Consider that $A$ is a symmetric operator on a Hilbert space, $\mathscr{H}$. If $v \in D(A^{n})$ for all $n$, then we say that $v \in C^{\infty}(A)$. A vector $v$ is called a \emph{semi-analytic vector} for $A$ if and only if
\begin{align}
	\sum_{n=0}^{\infty}{||A^{n}v|| \over (2n)!}t^{n} < \infty
\end{align}
for some $t>0$.
\end{defn}
The title is intended to denote similarity to the stronger condition of a $v$ an analytic vector of $A$ as given by Nelson \cite{Nelson}:
\begin{align}
	\sum_{n=0}^{\infty}{||A^{n}v|| \over n!}t^{n} < \infty
\end{align}
for some $t>0$. The theorem in question is as follows.
\begin{unthm}[Semi-Analytic Vector Theorem]
If $A$ is a symmetric operator of $\mathscr{H}$ which is bounded below so that $D(A)$ contains a set of semi-analytic vectors of $A$ which are dense in $\mathscr{H}$, then $A$ is essentially self-adjoint.
\end{unthm}

\begin{lem}\label{aplem01}
If $A > 0$ and $A$ has deficiency indices $[m,m]$ ($m < \infty$) then every self-adjoint extension of $A$ is semibounded.
\end{lem}

\begin{proof}
One may follow the arguments of \cite{AkGl}, p.\ 115-116. If $\widetilde{A}$ is any self-adjoint extension of $A$, then $D(\widetilde{A})/D(A)$ has dimension $m$ so that $\widetilde{A}$ has a spectral projection on $(-\infty,0)$ of dimension at most $m$.
\end{proof}

\begin{thm}\label{apthm01}
If $A > 0$ and has a unique semibounded self-adjoint extension, then $A$ is essentially self-adjoint.
\end{thm}

\begin{proof}
Suppose that the premise of the theorem is false and let $A$ have deficiency indices $[m,m]$. We must be aware of the case $m=\infty$. Let $A_{F}$ be the Friedrichs extension of $A$ \cite{F01}\cite{F02}. If $m \neq 0$, then we can find a symmetric operator $\widetilde{A}$ with deficiency indices $[1,1]$ so $A \subset \widetilde{A} \subset A_{F}$. One has that $A_{F} >$, $\widetilde{A} > 0$. Therefore by Lemma \ref{aplem01} one has that all other self-adjoint extensions of $\widetilde{A}$ are semibounded. It must then be the case that $A$ has more than one semibounded extension if $m \neq 0$. One may then conclude that $m = 0$, which is to say that $A$ is essentially self-adjoint.
\end{proof}

\begin{proof}[Proof of Semi-Analytic Vector Theorem]
Since $A$ is semibounded it has self-adjoint extensions by \cite{Neu}. For this case the Friedrichs extension exists. By Theorem \ref{apthm01} we need only show that $A$ has a unique semibounded self-adjoint extension. If $A$ has a dense set of semi-analytic vectors, then $A$ has at most one self-adjoint $\widetilde{A} > 0$. Let $v \in D(A)$ be a semi-analytic vector of $A$ and let $\mathrm{d}\mu_{\lambda,k}$ be the scalar spectral measure for $A$, where $0 \le k \le n$ is a multiplicity index for $\sigma(A)$. One has that
\begin{align}
	\sum_{k = 0}^{n}\int_{0}^{\infty} \lambda^{2m} |v(\lambda,k)|^{2} \mathrm{d}\mu_{\lambda,k} &= ||A^{m}v||^{2} < (c_{1}c_{2}^{m}(2m!))^{2} < c_{3}c^{m}(4m)! \\
	\Rightarrow\quad\sum_{m=0}^{\infty}\sum_{k=1}^{n}\int_{0}^{\infty} \lambda^{m/2}{t^{m} \over m!} |v(\lambda,k)| \mathrm{d}\mu_{\lambda,k} &< \infty \quad \text{for} \quad |t| < c^{-1/4},
\end{align}
where $\{c_{j}\}_{j=1}^{3}$ and $c$ are constants. By the Dominated Convergence Theorem one has that $\sum_{k=1}^{n}\int_{0}^{\infty} \exp(x^{1/2}t) |v(\lambda,k)| \mathrm{d}\mu_{\lambda,k} < \infty$ for $|t| < c^{-1/4}$. Since $|\cos y^{1/2}| < \exp |\Im y^{1/2}|$ one has that $\sum_{k=1}^{n}\int_{0}^{\infty} \cos(x^{1/2}t) |v(\lambda,k)| \mathrm{d}\mu_{\lambda,k} < \infty$ for $|\Im t| < c^{-1/4}$ and is therefore analytic in $t$, for $t$ in the strip $|\Im t| < c^{-1/4}$, and is given by the power series $\sum_{n=0}^{\infty} (2n!)^{-1}t^{n} (v,(-A)^{n}v)$ if $|t| < c^{-1/4}$. One in turn has that $(v,\cos(t\widetilde{A}^{1/2})v)$ is specified uniquely by $(v,A^{n}v)$ for real $t$ if $v$ is semi-analytic and $\widetilde{A}$ is a positive self-adjoint extension. If $A$ has a dense set of semi-analytic vectors, then $\cos(t\widetilde{A}^{1/2})$ is uniquely determined independently of the choice of self-adjoint extension. By spectral theorem one has
\begin{align}
	(\widetilde{A}+1)^{-1} = \int_{0}^{\infty} e^{-t} \cos(t\widetilde{A}^{1/2}) \mathrm{d}t
\end{align}
and therefore $\widetilde{A}$ is uniquely determined.

If $v$ is a semi-analytic vector for $A$, then it is also a semi-analytic vector for $A+x$, where $x$ is any positive real number. If $||A^{m}v|| < c_{1}c_{2}^{m}(2m)!$, then one has
\begin{align}
	|| (A + x)^{m}v || &\le \sum_{n=0}^{m}\binom{m}{n}x^{m-n}||A^{n}v|| \le c_{1}(2m)! \sum_{n=0}^{m}\binom{m}{n}x^{m-n}c_{2}^{n} \\
		&\le c_{1}(c_{2}+x)^{m}(2m)!.
\end{align}
Therefore the argument for uniqueness of $\widetilde{A}$ implies that an operator $A$ with a dense set of semi-analytic vectors has at most one extension $\widetilde{A}$ with $\widetilde{A} > -x$.
\end{proof}

\thanks{We thank Marius Beceanu for helpful discussions. This work was supported in part by the NSF grant DMS 1201394.}

\end{document}